\documentclass{article}

\usepackage{arxiv}
\usepackage{array,epsfig}
\usepackage{amsthm}
\usepackage{amsmath}
\usepackage{amsfonts}
\usepackage{color}
\usepackage[shortlabels]{enumitem}
\usepackage{booktabs}
\usepackage{sidecap}
\usepackage{setspace}
\usepackage{comment}
\usepackage{url}
\usepackage{float}
\usepackage{apacite}
\usepackage{parskip}
\usepackage{standalone}
\usepackage{tabu}
\usepackage{rotating}
\usepackage{multirow}
\usepackage[font=footnotesize,labelfont=bf]{caption}
\usepackage{bm, bbm, mdframed, mathtools, mathabx, multirow, subfigure, tikz, multirow, tcolorbox, soul, caption}
\usepackage[english]{babel}
\usepackage[utf8]{inputenc}
\usepackage[mathscr]{euscript}
\newcommand{\E}{\mathbb{E}}
\newcommand{\Var}{\mathrm{Var}}

\numberwithin{equation}{section} 
\tcbuselibrary{breakable}
\usepackage{mathdots}
\usepackage{lscape}
\usepackage{pdflscape}

\usepackage[round, longnamesfirst, authoryear, sort]{natbib}
\usepackage[autostyle]{csquotes}  
\usepackage[citecolor=black, linkcolor = black, colorlinks=true, allcolors=black]{hyperref}
\usepackage{glossaries}
\usepackage{appendix}
\usepackage{verbatim}
\usepackage{tcolorbox}
\usepackage{xcolor}
\usepackage{listings}
\usepackage{rotating}
\usepackage{comment}
\usepackage{graphicx}
\usepackage{dsfont}
\usepackage[ruled,vlined]{algorithm2e}
\usepackage{bookmark}
\usepackage{setspace}
\usepackage{xr}
\newtheorem{assumption}{Assumption}
\newtheorem{theorem}{Theorem}

\usepackage{threeparttable}
\usepackage{tikz}
\usetikzlibrary{shapes.geometric, arrows, positioning, arrows.meta, patterns}
\tikzstyle{c_solid} = [circle, minimum width=1cm, minimum height=1cm,text centered,draw=black, fill=white]
\tikzstyle{c_dashed} = [circle, minimum width=1cm, minimum height=1cm,text centered,draw=black, dashed]
\tikzstyle{arrow} = [thick,->,>={Stealth[scale=1.3]}]

\newcommand{\norm}[1]{\Big\lVert#1\Big\rVert}
\newcommand{\smallnorm}[1]{\lVert#1\rVert}



\title{Improving the Finite Sample Estimation of Average Treatment Effects using Double/Debiased Machine Learning with Propensity Score Calibration}

\date{\today}

\author{Daniele Ballinari\thanks{We thank
Michael Lechner for comments. The views, opinions, findings, and conclusions or recommendations expressed in this paper are strictly those of the author(s). They do not necessarily reflect the views of the Swiss National Bank (SNB). The SNB takes no responsibility for any errors or omissions in, or for the correctness of, the information contained in this paper.}\thanks{Editing was assisted by ChatGPT and Grammarly.}\\
    Swiss National Bank\\
    Börsenstrasse 15\\
    8001 Zurich, Switzerland\\
    \texttt{daniele.ballinari@snb.ch} \\
\And
Nora Bearth\\
    University of St. Gallen\\
    Rosenbergstrasse 22\\
    9000 St. Gallen, Switzerland\\
    \texttt{nora.bearth@unisg.ch}\\
}

\renewcommand{\headeright}{Working paper}
\renewcommand{\undertitle}{Working paper}
\renewcommand{\shorttitle}{Propensity Score Calibration for Double/Debiased Machine Learning}

\begin{document}
\maketitle

\begin{abstract}
In the last decade, machine learning techniques have gained popularity for estimating causal effects. One machine learning approach that can be used for estimating an average treatment effect is Double/debiased machine learning (DML) \citep{Chernozhukov:2018}. This approach uses a double-robust score function that relies on the prediction of nuisance functions, such as the propensity score, which is the probability of treatment assignment conditional on covariates. Estimators relying on double-robust score functions are highly sensitive to errors in propensity score predictions. Machine learners increase the severity of this problem as they tend to over- or underestimate these probabilities. Several calibration approaches have been proposed to improve probabilistic forecasts of machine learners. This paper investigates the use of probability calibration approaches within the DML framework. Simulation results demonstrate that calibrating propensity scores may significantly reduces the root mean squared error of DML estimates of the average treatment effect in finite samples. We showcase it in an empirical example and provide conditions under which calibration does not alter the asymptotic properties of the DML estimator.
\end{abstract}

\jelcodes{C01 \and C14 \and C21}

\newpage
\section{Introduction}

The increasing amount of observational data has led to a growing interest in causal inference methods that can handle high-dimensional data. As a result, machine learning techniques for estimating causal effects have gained popularity. A well-known example that can be used is the double/debiased machine learning (DML) estimator that relies on the prediction of nuisance functions with machine learning methods \citep*{Chernozhukov:2018}. One of the nuisance functions for estimating an average treatment effect (ATE) is the propensity score, which is the probability of receiving treatment conditional on a set of covariates. While an increasing number of covariates makes the unconfoundedness assumption in observational studies more plausible, propensity score estimates become more extreme as treatment assignment can be better predicted, and the overlap assumption is more difficult to satisfy \citep{DAmour:2021}.\footnote{The unconfoundedness and overlap assumptions are two crucial assumptions in a selection-on-observables setting.}

Extreme propensity scores have been identified to be problematic for estimators relying on double-robust score functions or inverse probability weighting already before the use of machine learners \citep*[e.g.][]{Froelich:2004, Busso:2014}. Methods of re-weighting the propensity scores or trimming methods have been developed and analysed to overcome this problem. For example, \cite*{Crump:2009} propose a rule of thumb of just discarding all propensity scores that are smaller than 0.1 and larger than 0.9. \cite{Yang:2018} provide asymptotic results for a smooth weighting approach approximating sample trimming. \cite*{Huber:2013} propose to trim and re-weight extreme inverse propensity scores. In the context of unbalanced treatment assignment where extreme propensity score estimates are likely, \cite{Ballinari:2024} proposes an estimation procedure for the ATE relaying on undersampling the dataset.


Machine learners aggravate the problem of extreme propensity scores, as their increased flexibility allows for even better label classification. In finite samples, however, researchers have shown that machine learners tend to over- or underestimate the probability of a label, i.e. the treatment probability \citep*[e.g.,][]{Niculescu:2005}. In fact, machine learners are generally designed to perform well in classification tasks and less so in quantifying the likelihood of a class. Deep neural networks, for example, have been found to overestimate probabilities in different settings \citep{Clarte:2023}. Errors in the predicted propensity scores can lead to biased finite sample estimates of the ATE. \cite*{Lechner:2024} have demonstrated through a large simulation study that using DML results in a biased ATE when the selection into treatment is high.

In a prediction setting, calibration methods have been employed to improve the performance of machine learners for classification and regression problems \citep*[e.g.][]{Bella:2010}. Calibration approaches aim to adjust the predicted probabilities closer to the true probabilities. In other words, a well-calibrated binary classifier should assign a predicted probability value (propensity score), such as 0.8, to a group of samples in such a way that roughly 80\% of those samples truly belong to this class (are treated).
Popular approaches to calibrate predicted probabilities are Platt scaling \citep{Platt:1999}, Beta scaling \citep{Kull:2017}, isotonic regression \citep{Zadrozny:2002}, Venn-Abers calibration \citep{Vovk:2012}, temperature scaling \citep{Guo:2017} and expectation consistent calibration \citep{Clarte:2023}.

This paper proposes to use calibration methods to improve the predictions of the propensity scores in the ATE estimation using the DML framework. The DML estimator is adjusted to include an additional estimation step for calibrating the propensity scores. We compare the performance of the ATE estimated with calibrated propensity scores to traditional DML and a re-weighting approach through a simulation study and show under what conditions calibration does not alter the asymptotic properties of DML. Our findings demonstrate that calibrating propensity scores may significantly reduce the root mean squared error (RMSE) of the ATE in finite samples, especially in scenarios with high treatment selectivity. The reduction in RMSE is primarily driven by a reduced bias in the ATE estimate, consequently improving the coverage of the confidence intervals. When there is almost no selectivity, traditional DML performs well, and calibration does not improve the performance but it does not harm the performance either. The higher the selection into treatment, the larger the improvement that can be achieved by calibrating the propensity score. Most of the calibration methods analysed in this paper perform well, with Venn-Abers calibration, Platt scaling, and Beta scaling being the best-performing calibration methods in the simulation study. We find that lower mean squared errors of the machine learner predicting the propensity score are associated with lower bias in the ATE. Practitioners are thus advised to choose the calibration method with the lowest Brier score, which measures the squared difference between the predicted probabilities and the treatment indicator (see Section \ref{subsec:simulation_design}). Last, we showcase the benefits of calibration in an empirical example evaluating the effect of language courses for unemployed individuals on their employment probability. 


The contribution of this paper is twofold. First, we contribute to a small but growing literature that uses calibration methods in the context of treatment effect estimation. \cite*{Gutman:2022} explore using calibrated propensity scores for the inverse probability weighting estimator. \cite*{VanderLaan:2024a} demonstrate the use of isotonic calibration for calibrating inverse propensity weights. Some papers combine calibration methods with the estimation of heterogeneous treatment effects. \cite*{vanderLaan:2023} propose a new nonparametric and doubly-robust method for calibrating predictors of heterogeneous treatment effects, but they do not calibrate the propensity scores. \cite*{Xu:2022} propose a method to assess if heterogeneous treatment effects are well calibrated. \cite{Shachi:2024} analyse the improvements in ATE estimation when using Platt scaling and isotonic regression for propensity scores predicted by a logistic regression. Second, we contribute to the small literature focusing on the finite sample performance of DML. \cite*{Knaus:2021} compare different machine learning estimator for heterogeneous treatment effects in a large simulation study. \cite*{Lechner:2024} show that DML is biased in finite samples when the selection into treatment is high. To the best of our knowledge, this is the first study to investigate the calibration of propensity scores in a DML framework and compare a comprehensive set of calibration methods and machine learners in a simulation study.

The remainder of the paper is structured as follows. Section \ref{sec:notation} defines the notation. Section \ref{sec:calibration_methods} discusses the different calibration methods used. In Section \ref{sec:estimation}, we explain the estimation procedure and present asymptotic results. Section \ref{sec:simulation_study} describes the simulation study and its results. In Section \ref{sec:empirical_application} we use the new procedure in an empirical application and last, Section \ref{sec:conclusion} concludes.

\section{Notation} \label{sec:notation}
We observe a dataset of $N$ i.i.d. observations and denote the index set as $\mathcal{I} = \{1, \dots, N\}$ and the dataset as $\mathcal{S} = \{(D_i, Y_i, X_i) : i \in \mathcal{I}\}$. This paper focuses on the case where $D_i$ is binary, taking values $d \in \{0,1\}$. However, the approach can easily be extended to discrete $D_i$. The random vector $X_i = [X_{i,1}, \dots, X_{i,q}]$ contains $q$ covariates. We denote the conditional expectation of the observed outcome as $\mu(d,x):= \E[Y_i|X_i=x, D_i=d]$ and the propensity score as $p(x) := P(D_i = 1 | X_i = x)$ and will refer to them as \emph{nuisance functions}. We are interested in the estimation of the expected difference in conditional expectations $\theta:=\E[\mu(1,X_i) - \mu(0,X_i)] = \E[\tau(D_i,X_i,Y_i)]$, where $\tau(d,x,y)$ is the efficient score function defined as:
\begin{equation}\label{eq:efficient_score}
    \tau(d,x,y) := \mu(1,x) -\mu(0,x) + \frac{d(y - \mu(1, x))}{p(x)} - \frac{(1-d)(y - \mu(0, x))}{1-p(x)}.
\end{equation}

Under suitable assumptions, the quantity of interest $\theta$ corresponds to the ATE: (i) conditional independence, (ii) common support, (iii) exogeneity of confounders, and (iv) stable unit treatment value assumption \citep[for more details we refer to, among others,][]{Imbens:2009}. The results in this paper do however not dependent on the identification of a treatment effect. Indeed, when these assumptions are not fulfilled, the quantity $\theta$ corresponds to the average prediction effect as described in \cite{Chernozhukov:2018}.

In what follows, we denote the $L_s$-norm as $\smallnorm{\cdot}_s$, for example $\smallnorm{g(X)}_s = \int\vert g(x)^s\vert^{1/s} dP(x)$. A sequence converging to zero is denoted by $o(1)$ and a bounded sequence by $O(1)$. Similarly, a sequence of random variable converging in probability to zero is denoted by $o_p(1)$ and a sequence bounded in probability by $O_p(1)$.

\section{Calibration Methods} \label{sec:calibration_methods}

\subsection{Definition}

We are interested in the estimation of the propensity score defined as $p(x) = P(D_i = 1 | X_i = x)$. Let $\hat p(x)$ denote the predicted probability coming from a statistical/machine learning model fitted on a training dataset $\mathcal{S}_T = \{(D_i, Y_i, X_i) : i \in \mathcal{I}_T\}$, where $\mathcal{I}_T \subset \mathcal{I}$. We would hope that $\hat p(x)$ can be interpreted as the probability of receiving treatment:
\begin{equation}\label{eq:ideal_propensity_score}
    \hat p(x) \approx \E[D_i|X_i = x] = p(x).
\end{equation}
Equation \eqref{eq:ideal_propensity_score} is satisfied if $\hat p(x)$ is (perfectly) calibrated, that is \citep{Gupta:2020}:
\begin{equation}
    \E[D_i\vert \hat p(X_i) = p] = p
\end{equation}
a.s. for all $p$ in the range of $\hat p$, where the expectation is taken over $D_i$ for $i\notin \mathcal{I}_T$.
In words, if the estimated propensity score is 0.8, we would expect that 80\% of the individuals with the same propensity score $p$ are actually treated. In essence, if the model is perfectly calibrated, the predicted probability should match the true probability, ensuring that, on average, the model's predictions align with reality \citep{Wang:2023}. Unfortunately, this is generally not the case without strong distributional assumptions, which may not hold in practice \citep{Gupta:2020}. We follow \cite{vanderLaan:2023} and quantify the calibration of the propensity score by the $L_2$ calibration error:
\begin{equation}\label{eq:calibration_error}
    CAL(\hat p) = \E\big[(\hat p(X_i) - p(X_i))^2 \vert \mathcal{S}_T\big].
\end{equation}

Several approaches have been proposed to improve the calibration of probability predictions \citep{Bella:2010}. A calibrator is a function $f: [0,1] \mapsto [0,1]$ mapping $\hat{p}$ to a predictor $\pi(x) := (f \circ \hat p) (x)$ with presumably improved calibration properties, that is $CAL(\pi) < CAL(\hat p)$. For the true propensity score $p(x)$, the calibrator corresponds to the identity function. In general, however, the calibrator function $f$ is not known and is typically estimated on a so called \emph{calibration dataset} $\mathcal{S}_C = \{(D_i, Y_i, X_i) : i \in \mathcal{I}_C\}$, where $\mathcal{I}_C \subset \mathcal{I}$ and $\mathcal{I}_C \cap \mathcal{I}_T = \emptyset $ (see Section \ref{sec:estimation} for more details). An estimated calibrator is denoted as $\tilde\pi(x)$, to emphasize that it is estimated on a different dataset than the propensity score $\hat p(x)$. In finite samples, perfect calibration is generally not achievable and many empirical calibration techniques restore to asymptotic calibration, where the calibration error vanishes as the sample size grows \citep{Gupta:2020}. In the reminder of this section we give an overview of several popular approaches introduced by the machine learning literature to estimate the calibrator map.

\subsection{Platt Scaling}
In the context of support vector machines, \cite*{Platt:1999} introduced an approach commonly referred to as Platt scaling or Platt calibration, where the calibration map is defined as a logistic transformation:
\begin{align}\label{eq:platt_scaling}
  \pi(x) = \frac{1}{1 + \exp(\beta \cdot \hat p(x) + \alpha)}.
\end{align}
The scalar parameters $\alpha$ and $\beta$ are estimated by maximum likelihood on the calibration dataset $\mathcal{S}_C$. The logistic mapping between the predicted and calibrated probabilities can be derived from the assumption that the propensity scores within each class (treatment) are normally distributed with the same variance \citep{Kull:2017}.\footnote{$P(\hat p(X_i) \leq c| D_i = d, \mathcal{S}_T)$ corresponds to the cumulative normal distribution function with $\mathcal{S}_T$ being the training set.} Due to the parametric assumption, the calibration error might not vanish asymptotically \citep{Gupta:2020}. In particular, the calibration error does not vanish asymptotically if the propensity scores $\hat p(x)$ are already perfectly calibrated since the identity function is not a member of the calibration map defined in Equation \eqref{eq:platt_scaling}. Platt scaling is a simple and computationally efficient method, which is particularly useful when the sample size is small \citep{vanderLaan:2023}.

\subsection{Beta Scaling}
The assumption that propensity scores within each class are normally distributed is unreasonable for many probabilistic classifiers that predict probabilities in the range $[0,1]$. \cite{Kull:2017} propose a more flexible form of Platt's scaling, called Beta scaling. Their approach can be derived by assuming that the predicted probabilities are beta-distributed within each class.\footnote{$P(\hat p(X_i) \leq c| D_i = d, \mathcal{S}_T)$ corresponds to the cumulative beta distribution function with $\mathcal{S}_T$ being the training set.}  The calibration map is defined as:
\begin{align}\label{eq:beta_scaling}
    \pi(x) = \frac{1}{1 + 1/e^\alpha \frac{\hat p(x)^{\beta_0}}{(1-\hat p(x))^{\beta_1}}}
\end{align}
with $\beta_0, \beta_1 \geq 0$. Beta scaling can be shown to be equivalent to a logistic regression on $\log(\hat p(x))$ and $\log(1-\hat p(x))$ \citep{Kull:2017}. We thus estimate the parameters $\alpha, \beta_0, \beta_1$ by maximum likelihood on the calibration dataset $\mathcal{S}_C$, with the constraint $\beta_0, \beta_1 \geq 0$. At the cost of an additional parameter, Beta scaling is more flexible than Platt scaling. In particular, the calibrator map does contain the identity function, parametrized by $\alpha=0, \beta_0=1, \beta_1=1$.

\subsection{Isotonic Calibration}\label{subsec:isotonic}
\cite{Zadrozny:2002} relax the parametric assumptions of Platt and Beta scaling by calibrating probabilities with a nonparametric approach: isotonic regression, also known as monotonic regression \citep*{Barlow:1972}. The idea is to fit a piecewise constant function to the predicted probabilities that is nondecreasing. For a new observation with covariates $X_i=x$ and raw predicted propensity score $\hat p(x)$, the calibrated probability is given by:
\begin{align*}
  \pi(x) = \begin{cases}
 \pi_{(1)} & \text{if } p(x) \leq \hat p_{(1)} \\
 \pi_{(i)} + \frac{\hat p(x) - \hat p_{(i)}}{\hat p_{(i + 1)} - \hat p_{(i)}}(\pi_{(i + 1)} - \pi_{(i)}) & \text{if } \hat p_{(i)} \leq \hat p(x) \leq \hat p_{(i + 1)} \\
 \pi_{(N_C)} & \text{if } \hat p(x) \geq \hat p_{(N_C)}
\end{cases}
\end{align*}
where $\hat p_{(i)}$ is the $i$-th ordered value of the set $\{\hat p(X_i) : i \in \mathcal{I}_C\}$ and $N_C$ is the cardinality of $\mathcal{S}_C$. The parameters of the isotonic regression are estimated over the calibration dataset by the following optimization problem:
\begin{align*}
 &\min_{\pi_{(1)}, \dots, \pi_{(N_C)}} \sum_{i\in\mathcal{I}_C} (D_i - \hat \pi(X_i))^2 \quad \text{s.t.} \; \pi(X_i)\leq \pi(X_j) \; \text{for all $i,j$ where} \; \hat p(X_i) \leq \hat p(X_j)
\end{align*}
which can be computed using the pool adjacent violators algorithm. Isotonic calibration is a popular method in the machine learning literature \citep{Zadrozny:2002, Gupta:2021}. \cite{vanderLaan:2023} show that the calibration error $CAL(\hat \pi(x))$ goes to zero as $N_C$ grows. The nonparametric nature of this calibrator comes at the cost of overfitting when the calibration dataset is small \citep{Caruana:2006}. Moreover, it is crucial to note that calibrated propensity scores using isotonic regression can take values of 0 and 1. In particular, it can be shown that when in the calibration dataset $\mathcal{S}_C$ the observation with the smallest uncalibrated predicted propensity score is not treated, then $\tilde\pi_{(1)}=0$. This leads to significant problems in a DML setting, as the inverse of the propensity score is used to estimate the ATE.

\subsection{Venn-Abers Calibration}
\cite{Vovk:2012} and \cite{Vovk:2015} propose a calibration approach based on Venn predictors. Their approach relies on the nonparametric isotonic regression (see Section \ref{subsec:isotonic}) but overcomes some of its shortcomings (e.g. overfitting). Given a calibration dataset $\mathcal{S}_C$ and a new observation $(D_i, Y_i, X_i)$, we proceed as follows. First, using an isotonic regression fit a calibrator $\pi_0(x)$ on the sample $\{(\hat p(X_i), D_i) : i \in \mathcal{I}_C\} \cup \{(\hat p(X_i), 0)\}$. A second calibrator $\pi_1(x)$ is fitted on the sample $\{(\hat p(X_i), D_i) : i \in \mathcal{I}_C\} \cup \{(\hat p(X_i), 1)\}$. \cite{Vovk:2015} provide guarantees that at least one of $\tilde\pi_0(x)$ and $\tilde\pi_1(x)$ is perfectly calibrated for the new observation. In other words, we fit an isotonic regression on the calibration dataset extended by a new observation. For this new observation, we once assume that it is treated and once that it is not treated. One of the two resulting probabilities will be perfectly calibrated for the new observation.\footnote{For computational efficiency, the isotonic regressions are not recomputed for each new observation $(D_i, Y_i, X_i)$. Instead, the above procedure is repeated on the sample $\{(\hat p(X_i), D_i) : i \in \mathcal{I}_C\} \cup \{(p, d)\}$ for a range of values $p$ and for $d\in\{0,1\}$. We obtain two vectors $\tilde\Pi_0$ and $\tilde\Pi_1$, storing the calibrated probabilities for the range of $p$ values. For a new observation $(D_i, Y_i, X_i)$, we select $\tilde\pi_0(x)$ and $\tilde\pi_1(x)$ from the vectors $\tilde\Pi_0$ and $\tilde\Pi_1$ according to its raw predicted probability $\hat p(x)$. This procedure is implemented in the \texttt{Python} package \texttt{venn-abers 1.4.4}.}

For the estimation of the ATE, we need, however, a single propensity score estimate. Following arguments of log loss minimization, \cite{Vovk:2012} show that a single calibrated probability can be obtained as:
\begin{equation}\label{eq:venn_abers}
 \tilde\pi(x) = \frac{\tilde\pi_1(x)}{1-\tilde\pi_0(x)+\tilde\pi_1(x)}.
\end{equation}
While the calibration guarantees no longer hold for Equation \eqref{eq:venn_abers}, $\tilde\pi(x)$ will be almost perfectly calibrated as long as $\tilde\pi_0(x)$ and $\tilde\pi_1(x)$ are close to each other. This directly follows from the fact that $\tilde\pi_0(x) \leq \tilde\pi(x) \leq \tilde\pi_1(x)$ and that either $\tilde\pi_0(x)$ or $\tilde\pi_1(x)$ are perfectly calibrated. In practice, the empirical findings of \cite{Vovk:2012} and \cite{Vovk:2015} show that Equation \eqref{eq:venn_abers} indeed leads to more calibrated probabilities compared to Platt scaling, isotonic regression or uncalibrated predictions.

\subsection{Temperature Scaling}
\cite{Guo:2017} propose a simple calibration method called temperature scaling. The idea is to rescale the log predicted probabilities by a scalar parameter $T$:
\begin{align}\label{eq:temperature_scaling}
  \pi(x) = \frac{\exp(\log(\hat p(x))/T)}{\exp(\log(\hat p(x))/T) + \exp(\log(1 - \hat p(x))/T)}.
\end{align}
For $T=1$, the calibrated probabilities are equal to the predicted probabilities, as $T\to 0$ the propensity score collapses to a point mass $\pi(x)= 1$ whereas for $T\to \infty$ the propensity score converges to 0.5. Temperature scaling is closely related to Platt and Beta scaling but is more parsimonious as it only has one parameter. \cite{Guo:2017} propose to estimate the parameter $T$ by maximizing the log-likelihood over the calibration dataset $\mathcal{S}_C$. Empirical results show that temperature scaling applied to neural networks is competitive with more complex calibration methods \citep{Guo:2017}.

\subsection{Expectation Consistent Calibration}
Expectation consistent calibration has been introduced by \cite{Clarte:2023} to calibrate predictions of neural networks. The idea is to rescale the probabilities in the same way as done by the temperature scaling approach (see Equation \eqref{eq:temperature_scaling}). Instead of maximizing the log-likelihood, \cite{Clarte:2023} fit the temperature parameter $T$ by enforcing that over the calibration dataset, the average calibrated propensity score coincides with the average proportion of correct treatment ``classifications''. Let $\hat{d}_i = \mathds{1}(\hat p(X_i) \geq 0.5)$ where $\mathds{1}(\cdot)$ is the indicator function, that is, we predict treatment if the propensity score is larger than 0.5. Then $T$ is chosen such that the following equality holds: 
\begin{align*}
 \frac{1}{N_C}\sum_{i \in \mathcal{I}_C} \max\{\tilde \pi(X_i), 1-\tilde \pi(X_i)\} = \frac{1}{N_C}\sum_{i = 1}^N \mathds{1}(\hat d_i = D_i)
\end{align*}
where $N_C$ is again the cardinality of the calibration dataset $\mathcal{S}_C$. The intuition of this approach is that when the machine learning algorithm correctly classifies a certain share of observations (accuracy), its ``confidence'' should, on average, be equal to this share \citep{Clarte:2023}.

\section{Calibrated Double/Debiased Machine Learning Estimator}\label{sec:estimation}

\subsection{Estimation Procedure}
Double/debiased machine learning is an estimation method that enables using machine learning methods for causal inference \citep{Chernozhukov:2018}. One causal estimand of interest that is possible to estimate using DML is the ATE. The idea is to estimate the ATE by using the efficient score function defined in Equation \eqref{eq:efficient_score} first introduced by \cite{Robins:1995}. The efficient score function is Neyman-orthogonal, meaning that small estimation errors in the nuisance parameters do not affect the estimation of the ATE \citep{Chernozhukov:2018}. Therefore, machine learners can be used to estimate the nuisance parameters despite having a small regularization bias. Furthermore, the nuisance parameters are estimated using cross-fitting to avoid overfitting. \cite{Chernozhukov:2018} provide asymptotic theory for the DML estimator, demonstrating its consistency and asymptotic normality, even when machine learners are used to predict the nuisance functions (see Assumption \ref{ass:dml}).

Algorithm \ref{alg:Calibrated-DR-Learner} shows the procedure for estimating an ATE ($\theta$) with the DML approach introduced by \cite{Chernozhukov:2018} directly extended by the calibration procedure.
The algorithm partitions the dataset $\mathcal{S}$ into $K$ folds, trains the nuisance functions on $K-1$ folds, and the resulting models are used to generate predictions for the remaining fold. This evaluation fold is further partitioned into $J$ sub-folds. The calibrator $\pi(x)$ is estimated on $J-1$ sub-folds using one of the approaches discussed in Section \ref{sec:calibration_methods}, while the pseudo-outcomes $\tau_i$ are computed on the remaining sub-fold. This calibration process is repeated $J$ times to compute pseudo-outcomes for each sub-fold. 
The entire procedure is repeated across all $K$ folds, such that each fold is used once as the evaluation/calibration set. The ATE is then estimated by averaging the pseudo-outcomes. Other approaches to split the dataset into training $\mathcal{S}_T = \{(X_i, D_i, Y_i) : i\in \mathcal{I}_T\}$, calibration $\mathcal{S}_C = \{(X_i, D_i, Y_i) : i\in \mathcal{I}_C\}$ and evaluation sets $\mathcal{S}_E = \{(X_i, D_i, Y_i) : i\in \mathcal{I}_E\}$ are possible, as long as the sets are disjoint.


\begin{algorithm}[!h]
    \setstretch{1.35}
\caption{\textsc{Calibrated DML}}\label{alg:Calibrated-DR-Learner}
\SetKwInOut{Input}{Input}
\SetKwInOut{Output}{Output}
\SetKwBlock{Beginn}{begin}{end}
\Input{Dataset $\mathcal{S} = \{(X_i, D_i, Y_i) : i \in \mathcal{I}\}$, number of cross-fitting splits $K$} 
\Output{$ATE$ estimate $\hat{\theta}$, $ATE$ estimate standard error $SE(\hat{\theta})$}
\Beginn{
    partition the index set $\mathcal{I}$ into $K$ disjoint folds $\mathcal{I}^{(1)}, \dots, \mathcal{I}^{(K)}$.\\
\For{k \textnormal{in} $\{1, \dots,K\}$}{
    set $\mathcal{I}_T := \mathcal{I} \setminus \mathcal{I}^{(k)}$\\
    partition the fold $\mathcal{I}^{(k)}$ into $J$ sub-folds $\mathcal{I}^{(k,1)}, \dots, \mathcal{I}^{(k,J)}$\\
            \textsc{Response Functions}: \\
            estimate $\hat \mu(1,x) = \hat \E[Y_i | D_i = 1, X_i = x]$ using $\{(X_i, Y_i) :  i \in \mathcal{I}_T \text{ and } D_i = 1\}$ \\
            estimate $\hat \mu(0,x) = \hat \E[Y_i | D_i = 0, X_i = x]$ using $\{(X_i, Y_i) :  i \in \mathcal{I}_T \text{ and } D_i = 0\}$ \\
            estimate $\hat p(x) = \hat P(D_i = 1 | X_i = x)$ using $\{(X_i, D_i) :  i \in \mathcal{I}_T\}$ \\
            \For {j \textnormal{in} $\{1,\dots, J\}$}{
                set $\mathcal{I}_E := \mathcal{I}^{(k,j)}$\\
                set $\mathcal{I}_C := \mathcal{I}^{(k)}\setminus\mathcal{I}^{(k,j)}$\\
                \textsc{Calibration}: \\
                estimate $\tilde \pi(x) = \tilde{f}(\hat p(x))$ in $\{(X_i, D_i) :  i \in \mathcal{I}_C \}$\\
                \textsc{Pseudo-Outcome}: \\
            compute $\tau_i = \hat{\mu}(1, X_i) - \hat{\mu}(0, X_i) + \frac{D_i(Y_i - \hat{\mu}(1, X_i))}{\tilde{\pi}(X_i)} - \frac{(1-D_i)(Y_i - \hat{\mu}(0, X_i))}{(1-\tilde{\pi}(X_i))}$ for all $i \in \mathcal{I}_E$\\
            }
            }
    \textsc{ATE}: \\
    compute estimate $\hat{\theta} = \frac{1}{N}\sum_{i \in \mathcal{I}} \tau_i$ \\
    compute standard errors $SE(\hat{\theta}) = \sqrt{\frac{1}{N}\sum_{i \in \mathcal{I}} \tau_i^2 - \hat{\theta}^2}$ \\
}
\end{algorithm}

\subsection{Asymptotic Results}

Calibration of propensity scores is expected to improve the finite sample performance of the DML estimator. In the following, we provide assumptions under which a calibrator does not change the asymptotic properties of the estimator. In particular, we show that it has the same asymptotic distribution as the DML estimator. This asymptotic result may serve as a guideline for researchers to choose a suitable calibration method, which not only improves the finite sample performance but also does not deteriorate the asymptotic properties of the estimator.

\begin{assumption}\label{ass:dml}
    Let the realization set $\Xi_N$ be a shrinking neighborhood of the true nuisance functions $\mu_d(x)$ and $p(x)$. Let $\mathcal{X}$ be the support of the random vector $X_i$. Define the rates:
    \begin{align*}
        r_{\mu_d,N}&:=\sup_{\mu_d^* \in \Xi_N} \norm{\mu_d^*(X) - \mu_d(X)}_2\\
        r_{p,N}&:=\sup_{p^* \in \Xi_N} \norm{p^*(X) - p(X)}_2
    \end{align*} 
    The following conditions hold: (i) $\text{Var}[Y_i|D_i=d, X_i=x]<\infty$ for all $x\in\mathcal{X}$ and $d\in\{0,1\}$, (ii) $\epsilon\leq p(x) \leq 1-\epsilon$ and $\epsilon\leq \hat{p}(x) \leq 1-\epsilon$ for $\epsilon>0$ and for all $x\in\mathcal{X}$, (iii) $\hat\mu_d, \hat{p}, \in \Xi_N$ with probability at least $1-o(1)$, (iv) for $w>2$ $\sup_{\mu_d^* \in \Xi_N} \smallnorm{\mu_d^*(X) - \mu_d(X)}_w<\infty$ and $\sup_{p^* \in \Xi_N} \smallnorm{p^*(X) - p(X)}_w<\infty$, (v) $r_{\mu_d,N}=o(1)$ and $r_{p,N}=o(1)$, (vi) $r_{\mu_d,N}\cdot r_{p,N}=o(N^{-1/2})$.
\end{assumption}

Assumption \ref{ass:dml} imposes the usual conditions for the DML estimator to be consistent and asymptotic normally distributed \citep{Chernozhukov:2018}. Converges rates satisfying Assumption \ref{ass:dml}-(vi) have been proven for a variety of machine learning algorithms, e.g. Lasso \citep*{Belloni:2013}, random forests \citep*{Wager:2015}, boosting \citep*{Luo:2016} or neural nets \citep*{Farrell:2021}. In addition, we require the calibrator to fullfil the following assumption.

\begin{assumption}\label{ass:calibrator}
    For $p^* \in \Xi_N$ let $(f \circ p^*)(X_i) = \E[D_i\vert p^* (X_i)]$ and denote by $\tilde{f}$ an estimator of $f$. Denote the realization set by $\Lambda_{N,p^*}$, which is a shrinking neighborhood of the true nuisance function $f$. Define the rate:
    \begin{align*}
        r_{f,N}&:= \sup_{p^* \in \Xi_{N}} \sup_{f^* \in \Lambda_{N,p^*}} \norm{(f^* \circ p^*) (X_i) - (f \circ p^*) (X_i)}_2
    \end{align*} 
    The following conditions hold: (i) $\tilde{f} \in \Lambda_{N,p^*}$ with probability at least $1-o(1)$ for all $p^* \in \Xi_N$, (ii) $\epsilon\leq (\tilde{f}\circ p^*)(x) \leq 1-\epsilon$ for $\epsilon>0$, for all $x\in\mathcal{X}$ and for all $p^* \in \Xi_N$, (iii) $r_{f,N}=o(1)$, (iv) $r_{\mu_d,N} \cdot r_{f,N}=o(N^{-1/2})$, (v) $\sup_{p^* \in \Xi_N} \smallnorm{(f \circ p^*) (X_i) - p(X_i)}_2 \leq C \sup_{p^* \in \Xi_N} \smallnorm{ p^* (X_i) - p(X_i)}_2$ for fixed strictly positive constant $C$.
\end{assumption}

Assumption \ref{ass:calibrator} requires the calibrator to converge to the true expected value of the treatment $D_i$ conditional on the estimated of propensity score $\hat{p}(X_i)$. In particular, Assumption \ref{ass:calibrator}-(iv) requires the calibrator to converge fast enough. For example, when the estimated uncalibrated propensity scores are beta distributed within the treated and nontreated observations, Beta scaling converges at the parametric rate $N^{1/2}$, fulfilling the requirement of Assumption \ref{ass:calibrator}-(iv). Rates for the nonparametric calibrators using isotonic regressions can be found in, e.g. \cite{Zhang:2002}. Notice that the above assumption requires the calibrator to nest the identity function: the true propensity score $p \in \Xi_N$ and Assumption \ref{ass:calibrator} has to hold for all $p^* \in \Xi_N$. In fact, Platt scaling does not fulfill this requirement and the asymptotic result presented below do not hold for this calibrator.  Assumption \ref{ass:calibrator}-(v) requires the calibrator to increase the $L_2$-error of the uncalibrated propensity score estimator by no more than a finite factor $C$. The theoretical result is then given by the following theorem.

\begin{theorem}
    \label{theorem:convergence}
    Under Assumptions \ref{ass:dml} and \ref{ass:calibrator} it holds that:
    \begin{equation*}
        \sqrt{N} \big( \hat{\theta} - \theta \big) \overset{d}{\longrightarrow} \mathcal{N}(0, V)
    \end{equation*}
    where:
    \begin{equation*}
        V = \Var[\mu_1(X_i)-\mu_0(X_i)] + \E\bigg[\frac{\sigma_1^2(X_i)}{p(X_i)}\bigg] + \E\bigg[\frac{\sigma_0^2(X_i)}{1-p(X_i)}\bigg]
    \end{equation*}
    with $\sigma^2_d(X_i) = \Var[Y_i^d\vert X_i]$.
\end{theorem}

The proof of Theorem \ref{theorem:convergence} is relegated to Appendix \ref{app:asymptotic_proofs}. Theorem \ref{theorem:convergence} shows that the DML estimator with calibrated propensity scores has the same asymptotic distribution as the usual DML estimator. In particular, the estimator attains the parametric rate of convergence $\sqrt{N}$ and its variance achieves the semi-parametric efficiency bound \citep{Hahn:1998}. While this asympotic result does not shed light on the finite sample properties, it provides assumptions under which a calibration method does not alter the asymptotic distribution of the ATE estimator.

\section{Simulation Study} \label{sec:simulation_study}
\subsection{Data Generating Process (DGP)}
We generate six different DGPs to test the finite sample performance of the calibration methods. In the spirit of \cite{Nie:2021}, data is generated as follows:
\begin{align*}
 X_i &\sim \text{Unif}(0,1)^{30}, \quad D_i \mid X_i \sim \text{Bernoulli}(p(X_i)), \quad \epsilon_i \sim \mathcal{N}(0,1), \quad
 Y_i = b(X_i) + (D_i - 0.5) \times \frac{X_{i,1} + X_{i,2}}{2} + \epsilon_i
\end{align*}
We vary the difficulty of the baseline main effect and the propensity score function. Let $X_{i,j}$ denote the $j$th entry of the random vector $X_i$. An easy and a difficult baseline main effect $b(X_i)$ are generated as follows:
\begin{align*}
 b(X_i) = \begin{cases}
\begin{array}{ll}
 X_{i,1}\times X_{i,2} + 2 \times (X_{i,3} - 0.5)^2 + X_{i,4} + 0.5X_{i,5} & \text{if easy} \\
 \sin(\pi \times X_{i,1} \times X_{i,2}) + 2 \times (X_{i,3} - 0.5)^2 + X_{i,4} + 0.5X_{i,5} & \text{if difficult}
\end{array}
 \end{cases}
\end{align*}

where the difficult baseline main effect is the scaled \cite{Friedman:1991} function. We consider three different propensity score functions:
\begin{align*}
 p(X_i) = \begin{cases}
\begin{array}{ll}
 \frac{1}{1 + \exp(X_{i,1} - X_{i,2})} & \text{if easy} \\
0.1 + 0.6\beta_{2,4}(\min\{X_{i,1}, X_{i,2}\}) & \text{if difficult} \\
 0.05 + 0.9\beta_{2,4}(\min\{X_{i,1}, X_{i,2}\}) & \text{if difficult and extreme}
\end{array}
 \end{cases}
\end{align*}
where $\beta_{2,4}(x)$ is the beta cumulative distribution function with shape parameters 2 and 4. The first propensity score function has a logistic functional form and should be easy to estimate. Following \cite{Kunzel:2019}, the second and third propensity score functions are highly nonlinear and pose greater challenges for accurate estimation. Furthermore, the third function produces extreme propensity scores.
Figure \ref{fig:propensity_scores} shows the overlap of the propensity scores by plotting kernel densities of the propensity score distributions conditional on $D_i$.\footnote{Figure \ref{fig:propensity_scores} is based on a simulated sample of 2,000 observations. The kernel densities slightly overestimate (underestimate) the maximum (minimum) propensity score.} The first subplot has a nearly perfect overlap between the treated and untreated group, the second one shows already a strong selection into treatment, but the minimum and maximum probability of being treated lie at 10\% and 70\%, respectively. For the extreme propensity scores, the maximum probability of being treated lies at 95\% and the minimum probability at 5\%. Hence, the treatment selection is very strong.

\begin{figure}[h!]
    \centering
    \begin{minipage}{\textwidth}
        \caption{Overlap for the three different propensity scores}
    \label{fig:propensity_scores}
    \includegraphics[width = \textwidth]{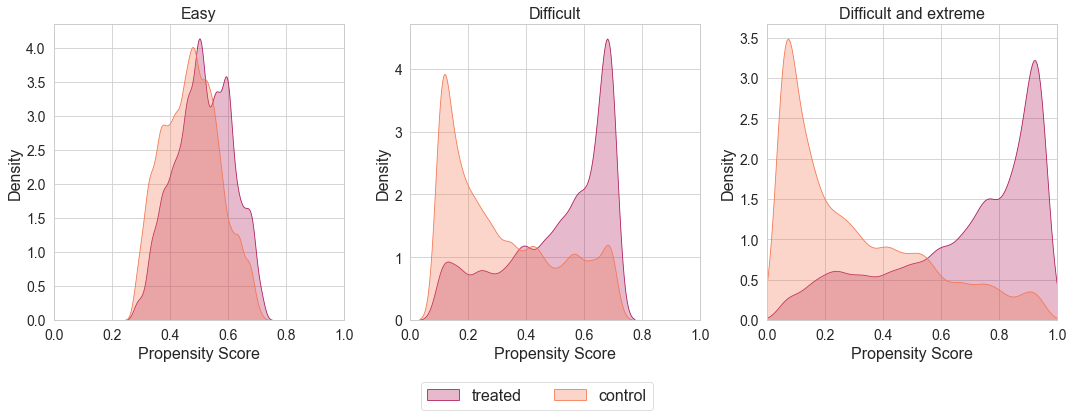}
    \scriptsize \textit{Note:} The figure depicts kernel density estimates of the propensity score distribution conditional on $D_i$ for a random simulated sample of 2,000 observations. The left chart shows the overlap for the easy propensity score, the middle chart for the difficult propensity score and right chart for the difficult and extreme propensity score.
    \end{minipage}
\end{figure}

We consider six different combinations of the baseline main effect and the propensity score function, resulting in six different DGPs as shown in Table \ref{overview_dgps}.\footnote{The Python code for the simulation study is available at \url{https://github.com/dballinari/Improving-the-Finite-Sample-Performance-of-DML-with-Propensity-Score-Calibration}.}

\renewcommand{\arraystretch}{1.1} 

\begin{table}[h!]
    \centering
    \caption{Overview of DGPs}
    \label{overview_dgps}
    \begin{adjustbox}{max width=\textwidth}
        \begin{threeparttable}
            \begin{tabular}{cc|ccc} \toprule
                & & \multicolumn{3}{c}{\textbf{Propensity Score}} \\ 
                & & Easy $p(x)$ & Difficult $p(x)$ & Extreme $p(x)$ \\ \midrule
                \multirow{2}{*}{\textbf{Baseline}} & \multicolumn{1}{c|}{Easy $b(x)$} & DGP 1 & DGP 2 & DGP 5 \\
                & \multicolumn{1}{c|}{Difficult $b(x)$} & DGP 3 & DGP 4 & DGP 6 \\ \bottomrule
            \end{tabular}
            \begin{tablenotes}[flushleft]
                \small
                \item \textit{Note:} This table presents the combinations of baseline effects and propensity scores used in different DGPs.
            \end{tablenotes}
        \end{threeparttable}
    \end{adjustbox}
\end{table}
\renewcommand{\arraystretch}{1}

\subsection{Simulation Design and Estimators}\label{subsec:simulation_design}
For each of the DGPs, we generate $R=1,000$ samples of sizes ($N$) 2,000, 4,000 and 8,000 and compare three different estimators for the nuisance functions, namely a random forest, gradient boosting, and a linear (for the conditional outcomes) and logistic (for the propensity scores) regression with $L_1$ penalty (Lasso). For the linear and logistic regressions, additional features are generated by including polynomials and interactions of order two of the raw covariates $X_i$. For each DGP and sample size, we tune the machine learner of choice on 20 simulation replications, choosing the hyperparameter combination that most frequently yields the best cross-validation performance. Table \ref{grid_search} and table \ref{selected hyperparameters} in Appendix \ref{app:grid_search} shows the grid search and selected hyperparameters used for the different estimators.


The different calibrated versions of the DML estimator are compared with the DML estimator \citep*{Chernozhukov:2018} and the DML estimator with normalized and truncated weights \citep*{Huber:2013}. The calibration methods used are Platt scaling, Beta scaling, isotonic regression, Venn-Abers calibration, temperature scaling and expectation consistent calibration. The performance of the estimators is evaluated using the root mean squared error (RMSE), the bias, the standard deviation and the coverage of the 95\% confidence interval of the ATE estimate $\hat\theta$. The performance of the propensity score estimation is evaluated using the Brier score \citep{Brier:1950} defined as:
\begin{equation}\label{eq:brier_score}
 \frac{1}{N}\sum_{i = 1}^N (\hat p(X_i) - D_i)^2
\end{equation}
and for calibrated propensity scores we replace $\hat p(X_i)$ with $\tilde{\pi}(X_i)$. Ideally we would measure the mean squared difference between the true propensity score and the estimated/calibrated propensity score. However, in empirical applications the true propensity score is not observed. The Brier score replaces the true unknown propensity score with the observed treatment indicator. At the population level, $p(x)$ minimizes the Brier score.\footnote{The solution to the optimization problem $\min_g \E[(g(X)-D)^2]$ is $g(X) = E[D|X] = p(X)$.} This score is a popular measure for assessing the performance of probability forecasts \citep[e.g.][]{Dal:2015}. Moreover, Equation \eqref{eq:brier_score} is a natural candidate to evaluate the performance of the propensity score estimation, as the theoretical results of the DML framework rely on the mean squared convergence of the propensity score estimate \citep{Chernozhukov:2018}.

\subsection{Results} 

Table \ref{Results_simulation_2000} shows the results for samples of size $N = 2,000$ for the different calibration methods, DGPs and machine learners. The results for the other sample sizes can be found in the Online Appendix \ref{app:additional_results}. By comparing the different DGPs, we see that DML performs well if the propensity score has a nearly perfect overlap and is easy to estimate (DGP 1 and 3). In this situation, using calibration methods does not improve the results but does also not worsen them. Moreover, all three machine learners work in these cases, as a linear model can easily approximate the data.

When the propensity score is challenging to estimate (DGP 2 and 4), the RMSE of the DML estimator increases, primarily due to a higher bias. Calibration methods can considerably improve the ATE estimation in terms of RMSE, with Venn-Abers calibration, Platt, or Beta scaling showing the best performance. As long as the outcome regression is easy to estimate and linear (DGP 2), all estimators perform well, with Lasso being the most effective. This underscores the practical relevance of the double-robustness property of the DML estimator: despite the difficulty in estimating the propensity score, the simple baseline main effect can be well estimated by a linear model. 

Lasso performs significantly worse than the other two machine learners for highly nonlinear baseline main effects (DGP 4). Calibration methods can greatly improve the results, reducing the RMSE for all three machine learners. In particular, for Lasso Venn-Abers calibration reduces the RMSE of DML by almost 60\%, reaching nearly the same RMSE as more sophisticated machine learners. Calibration methods reduce the bias in the ATE estimation, which is the main reason for the RMSE improvements. After calibration, the coverage of the confidence interval is close to the nominal level of 95\% for all calibration methods. These patterns become even more noticeable when the propensity scores become more extreme (DGP 5 and 6). Except for isotonic regression, all calibration methods reduce the RMSE of the DML estimator by at least 50\% when the propensity scores are extreme and the outcome regression is difficult to estimate (DGP 6). The reduction stems again from a decrease in the bias.

In summary, calibration of the propensity scores can significantly reduce the RMSE of the DML estimator when the propensity scores are difficult to estimate, regardless of the machine learner used. 
The improvements are mainly related to a reduction in the bias of the ATE estimate, while the standard deviation remains relatively stable. The reduction in bias translates into better coverage of the 95\% confidence interval. Reweighed DML improves only marginally the DML estimator's RMSE, indicating that the calibration of the propensity scores does not simply tackle the problem of extremely small or large propensity scores. Venn-Abers calibration, Platt scaling, and Beta scaling perform best among the calibration methods. Isotonic regression performs poorly because some propensity scores are pushed to zero or one, leading to a significant bias. Therefore, we exclude it from the discussion below.

\begin{table}[h!]
    \centering
    \caption{Simulation results for sample size $N = 2,000$}
\label{Results_simulation_2000}
    \begin{adjustbox}{max width=\textwidth}
        \begin{threeparttable}
    \begin{tabular}{l|cccc|cccc|cccc} \toprule
        Approach  &\multicolumn{4}{c}{Random Forest} &\multicolumn{4}{c}{Gradient Boosting} &\multicolumn{4}{c}{Lasso} \\  \midrule
        & RMSE & Bias & St. Dev. & Coverage & RMSE & Bias & St. Dev. & Coverage & RMSE & Bias & St. Dev. & Coverage \\ \midrule
        &\multicolumn{12}{c}{DGP 1 (easy outcome regression / easy propensity scores)} \\  \midrule
        DML & 0.050 & -0.015 & 0.047 & 0.937 & 0.049 & -0.009 & 0.048 & 0.947 & 0.048 & -0.005 & 0.048 & 0.946\\
Reweight DML & 0.050 & -0.015 & 0.047 & 0.937 & 0.049 & -0.009 & 0.048 & 0.947 & 0.048 & -0.005 & 0.048 & 0.946\\
Platt Scaling & 0.048 & -0.008 & 0.048 & 0.941 & 0.049 & -0.009 & 0.048 & 0.946 & 0.048 & -0.006 & 0.047 & 0.945\\
Beta Scaling & 0.048 & -0.008 & 0.048 & 0.941 & \textbf{0.049} & -0.008 & 0.048 & 0.945 & \textbf{0.048} & -0.006 & 0.047 & 0.944\\
Expectation Consistent & 0.049 & -0.010 & 0.048 & 0.941 & 0.049 & -0.009 & 0.048 & 0.944 & 0.048 & -0.006 & 0.047 & 0.945\\
Temperature Scaling & 0.049 & -0.010 & 0.048 & 0.942 & 0.049 & -0.009 & 0.048 & 0.946 & 0.048 & -0.006 & 0.047 & 0.943\\
Isotonic Regression & 1.216 & 0.204 & 1.199 & 0.929 & 1.281 & 0.224 & 1.262 & 0.939 & 1.294 & 0.122 & 1.288 & 0.936\\
Venn Abers & \textbf{0.048} & -0.005 & 0.048 & 0.942 & 0.049 & -0.006 & 0.049 & 0.943 & 0.048 & -0.004 & 0.048 & 0.943\\
 \midrule
        &\multicolumn{12}{c}{DGP 2 (easy outcome regression / difficult propensity scores)} \\  \midrule
        DML & 0.075 & 0.052 & 0.053 & 0.842 & 0.067 & 0.040 & 0.054 & 0.877 & 0.071 & 0.049 & 0.051 & 0.848\\
Reweight DML & 0.073 & 0.050 & 0.053 & 0.852 & 0.066 & 0.038 & 0.054 & 0.879 & 0.070 & 0.047 & 0.051 & 0.852\\
Platt Scaling & 0.059 & 0.021 & 0.055 & 0.938 & 0.057 & 0.010 & 0.056 & 0.934 & 0.055 & -0.007 & 0.054 & 0.942\\
Beta Scaling & 0.060 & 0.016 & 0.058 & 0.945 & \textbf{0.057} & 0.008 & 0.056 & 0.932 & \textbf{0.054} & -0.007 & 0.054 & 0.947\\
Expectation Consistent & 0.063 & 0.026 & 0.057 & 0.930 & 0.063 & 0.004 & 0.063 & 0.951 & 0.091 & -0.058 & 0.071 & 0.869\\
Temperature Scaling & 0.063 & 0.028 & 0.056 & 0.922 & 0.058 & 0.015 & 0.056 & 0.929 & 0.065 & -0.029 & 0.059 & 0.930\\
Isotonic Regression & 1.481 & -0.631 & 1.340 & 0.924 & 0.890 & -0.193 & 0.869 & 0.934 & 1.687 & -0.924 & 1.412 & 0.909\\
Venn Abers & \textbf{0.059} & 0.017 & 0.056 & 0.940 & 0.058 & 0.006 & 0.057 & 0.938 & 0.058 & -0.004 & 0.057 & 0.942\\
 \midrule
        &\multicolumn{12}{c}{DGP 3 (difficult outcome regression / easy propensity scores)} \\  \midrule
        DML & 0.050 & -0.014 & 0.048 & 0.930 & 0.050 & -0.009 & 0.049 & 0.944 & 0.049 & -0.007 & 0.049 & 0.950\\
Reweight DML & 0.050 & -0.014 & 0.048 & 0.930 & 0.050 & -0.009 & 0.049 & 0.944 & 0.049 & -0.007 & 0.049 & 0.951\\
Platt Scaling & 0.049 & -0.007 & 0.048 & 0.948 & 0.050 & -0.009 & 0.049 & 0.948 & 0.049 & -0.009 & 0.049 & 0.949\\
Beta Scaling & 0.049 & -0.008 & 0.048 & 0.949 & 0.049 & -0.009 & 0.049 & 0.948 & \textbf{0.049} & -0.009 & 0.049 & 0.948\\
Expectation Consistent & 0.049 & -0.009 & 0.048 & 0.947 & 0.050 & -0.009 & 0.049 & 0.947 & 0.049 & -0.009 & 0.049 & 0.947\\
Temperature Scaling & 0.049 & -0.009 & 0.048 & 0.946 & 0.050 & -0.010 & 0.049 & 0.949 & 0.049 & -0.009 & 0.049 & 0.948\\
Isotonic Regression & 1.237 & 0.230 & 1.215 & 0.930 & 1.293 & 0.239 & 1.270 & 0.931 & 1.309 & 0.070 & 1.308 & 0.944\\
Venn Abers & \textbf{0.049} & -0.005 & 0.049 & 0.944 & \textbf{0.049} & -0.006 & 0.049 & 0.947 & 0.049 & -0.006 & 0.049 & 0.950\\
 \midrule
        &\multicolumn{12}{c}{DGP 4 (difficult outcome regression / difficult propensity scores)} \\  \midrule
        DML & 0.081 & 0.059 & 0.055 & 0.808 & 0.086 & 0.068 & 0.053 & 0.760 & 0.141 & 0.131 & 0.052 & 0.282\\
Reweight DML & 0.079 & 0.056 & 0.055 & 0.828 & 0.083 & 0.063 & 0.053 & 0.786 & 0.137 & 0.127 & 0.052 & 0.307\\
Platt Scaling & 0.061 & 0.021 & 0.057 & 0.941 & 0.057 & 0.013 & 0.056 & 0.938 & 0.097 & 0.079 & 0.055 & 0.697\\
Beta Scaling & 0.061 & 0.014 & 0.060 & 0.949 & \textbf{0.056} & 0.009 & 0.056 & 0.939 & 0.080 & 0.058 & 0.055 & 0.816\\
Expectation Consistent & 0.065 & 0.026 & 0.059 & 0.935 & 0.065 & 0.007 & 0.064 & 0.947 & 0.131 & 0.112 & 0.067 & 0.636\\
Temperature Scaling & 0.064 & 0.028 & 0.058 & 0.928 & 0.061 & 0.025 & 0.056 & 0.925 & 0.120 & 0.105 & 0.058 & 0.559\\
Isotonic Regression & 1.530 & -0.657 & 1.382 & 0.928 & 0.983 & -0.423 & 0.888 & 0.916 & 1.516 & -0.048 & 1.515 & 0.937\\
Venn Abers & \textbf{0.059} & 0.015 & 0.057 & 0.943 & 0.057 & 0.006 & 0.056 & 0.944 & \textbf{0.059} & 0.009 & 0.058 & 0.946\\
 \midrule
        &\multicolumn{12}{c}{DGP 5 (easy outcome regression / difficult and extreme propensity scores)} \\  \midrule
        DML & 0.122 & 0.109 & 0.054 & 0.465 & 0.145 & 0.135 & 0.054 & 0.284 & 0.117 & 0.105 & 0.052 & 0.470\\
Reweight DML & 0.117 & 0.104 & 0.054 & 0.515 & 0.144 & 0.134 & 0.054 & 0.288 & 0.116 & 0.103 & 0.052 & 0.490\\
Platt Scaling & \textbf{0.072} & 0.035 & 0.063 & 0.916 & 0.077 & 0.046 & 0.062 & 0.889 & 0.165 & -0.070 & 0.150 & 0.922\\
Beta Scaling & 0.073 & 0.029 & 0.067 & 0.930 & 0.077 & 0.045 & 0.062 & 0.893 & 0.102 & -0.030 & 0.097 & 0.941\\
Expectation Consistent & 0.073 & 0.027 & 0.068 & 0.933 & 0.130 & -0.005 & 0.130 & 0.947 & 0.182 & -0.068 & 0.168 & 0.939\\
Temperature Scaling & 0.073 & 0.030 & 0.066 & 0.935 & 0.081 & 0.052 & 0.061 & 0.868 & 0.170 & -0.070 & 0.155 & 0.924\\
Isotonic Regression & 1.607 & -0.784 & 1.403 & 0.903 & 0.342 & 0.011 & 0.342 & 0.920 & 1.720 & -0.934 & 1.444 & 0.905\\
Venn Abers & 0.074 & 0.038 & 0.063 & 0.905 & \textbf{0.074} & 0.041 & 0.062 & 0.899 & \textbf{0.070} & 0.013 & 0.069 & 0.947\\
 \midrule
        &\multicolumn{12}{c}{DGP 6 (difficult outcome regression / difficult and extreme propensity scores)} \\  \midrule
        DML & 0.145 & 0.133 & 0.058 & 0.370 & 0.199 & 0.191 & 0.056 & 0.075 & 0.159 & 0.147 & 0.059 & 0.294\\
Reweight DML & 0.136 & 0.123 & 0.058 & 0.440 & 0.198 & 0.190 & 0.056 & 0.080 & 0.157 & 0.145 & 0.059 & 0.303\\
Platt Scaling & \textbf{0.065} & 0.011 & 0.064 & 0.948 & 0.070 & 0.031 & 0.063 & 0.927 & 0.285 & 0.212 & 0.190 & 0.838\\
Beta Scaling & 0.069 & 0.005 & 0.069 & 0.945 & 0.070 & 0.030 & 0.063 & 0.928 & 0.163 & 0.116 & 0.115 & 0.859\\
Expectation Consistent & 0.070 & 0.000 & 0.070 & 0.946 & 0.167 & -0.039 & 0.162 & 0.940 & 0.326 & 0.231 & 0.230 & 0.877\\
Temperature Scaling & 0.068 & 0.005 & 0.067 & 0.947 & 0.082 & 0.051 & 0.064 & 0.870 & 0.300 & 0.224 & 0.200 & 0.838\\
Isotonic Regression & 1.859 & -1.130 & 1.477 & 0.869 & 0.386 & -0.063 & 0.381 & 0.927 & 1.515 & 0.094 & 1.512 & 0.949\\
Venn Abers & 0.067 & 0.019 & 0.065 & 0.942 & \textbf{0.067} & 0.022 & 0.063 & 0.934 & \textbf{0.075} & 0.016 & 0.073 & 0.938\\
 \midrule
    \end{tabular}
    \begin{tablenotes}[flushleft]
     \item \textit{Note:} The table shows the results for the different calibration methods for the different DGPs with $N = 2000$. The RMSE is the root mean squared error, the bias is the average difference between the true value and the estimated value, the standard deviation is the standard deviation of the estimated values and the coverage is the coverage of the 95\% confidence interval.
    \end{tablenotes}
\end{threeparttable}
\end{adjustbox}
\end{table}

Next, we will focus on the results for DGP 4 for two reasons. First, as it is a challenging DGP, DML exhibits some bias, indicating room for improvement. As seen earlier, if the DGP is easy to estimate, calibration is unnecessary. Second, DGP 4 is still realistic, as we still have overlap, and there are not too extreme values in the propensity scores. Figure \ref{fig:rmse_mode_4} shows the RMSE for DGP 4 for the different machine learners and different sample sizes. For random forest and gradient boosting, the RMSE is of similar magnitude across the different calibration approaches. The difference in RMSE to the DML estimator and the reweighed DML estimator decreases with increasing sample size, meaning that calibration is more important for smaller sample sizes. As discussed earlier, the linear model performs considerably worse than the other two machine learners for the highly nonlinear DGP 4. Nevertheless, using Venn-Abers calibrated propensity scores significantly improves the RMSE of the Lasso estimator, reducing it to a similar level as the other two machine learners. This result holds across all three sample sizes.

\begin{figure}[h!]
    \centering
    \begin{minipage}{\textwidth}
        \caption{RMSE for different sample sizes (DGP 4)}
    \label{fig:rmse_mode_4}
    \includegraphics[width = \textwidth]{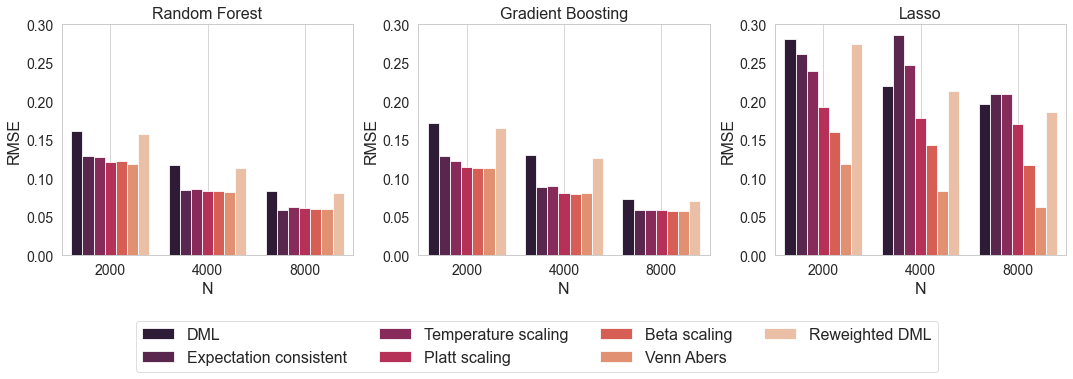}
    \scriptsize \textit{Note:} The figure depicts the RMSE across replications. From left to right, the charts show the results for gradient boosting, random forest and Lasso. The different colors represent the different calibration methods. The charts use 1,000 replications from samples of size 2,000, 4,000 and 8,000 generate from DGP 4 (difficult outcome regression / difficult propensity score).
    \end{minipage}
\end{figure}

The calibration methods improve the RMSE by improving the estimation of the propensity scores. This can be seen from Figure \ref{fig:rmse_bias_relation_mode_4}, which shows for each replication the relation between the Brier score of the propensity score and the bias of the ATE estimate for samples of size 4,000 drawn from DGP 4.\footnote{The same patterns can be seen for the relation between the bias of the ATE estimate and the mean squared difference between true and predicted propensity scores.} For all three machine learners, Platt scaling and Venn-Abers calibration reduces the Brier score of the propensity score. This reduction in the Brier score of the propensity score translates into a reduction of the ATE estimate's bias. However, the reduction in the ATE bias is not of the same proportion as the reduction in the Brier score of the propensity score since prediction errors of the conditional outcome also influence the ATE estimator. This can be seen from the results of the linear model, where despite a large reduction in the Brier score of the propensity score, we still have a considerable variation in the bias of the ATE estimate. In summary, the Brier score could guide researchers in their choice of the calibration method (if any). For example, if we use the Brier score as a criterion for choosing the calibration method in each replication sample generated from DGP 4 ($N=4,000$), the RMSE of estimates based on gradient boosting predictions would have amounted to 0.0016, a 60\% decrease compared to the DML estimator.\footnote{In more detail, for each replication we calibrate the predicted propensity scores using one of the calibration methods presented in Section \ref{sec:calibration_methods} (excluding the isotonic regression). The calibration procedure follows the same cross-fitting approach outlined in Algorithm \ref{alg:Calibrated-DR-Learner}. We then compute the Brier score for each calibration method and choose the calibration method with the lowest Brier score. The ATE estimate for the replication is computed using the propensity scores calibrated with the chosen calibration method.}

\begin{figure}[h!]
    \centering
    \begin{minipage}{\textwidth}
        \caption{Relationship between prediction of propensity score and ATE bias for DGP 4 and $N = 4,000$}
    \label{fig:rmse_bias_relation_mode_4}
    \includegraphics[width = \textwidth]{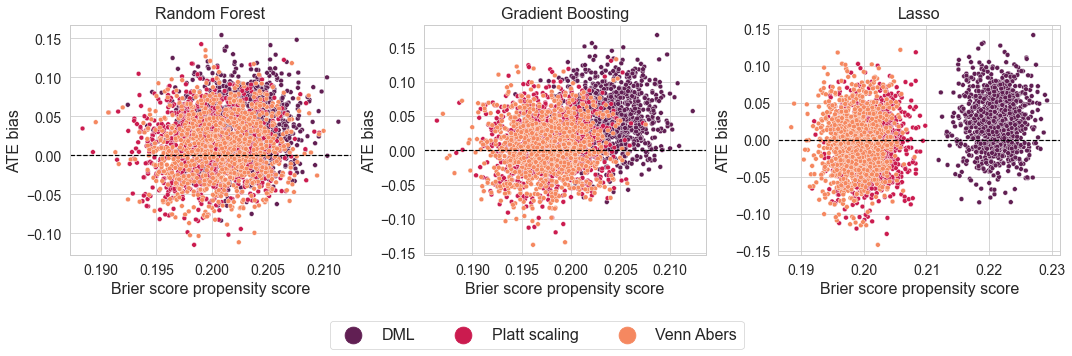}
    \scriptsize \textit{Note:} The figures show the relation between the ATE estimate's bias with the Brier score, defined as $\frac{1}{N}\sum_{i = 1}^N (\hat p(X_i) - D_i)^2$. From left to right, the charts show the results for gradient boosting, random forest and Lasso. The different colours show the different calibration methods. The charts use 1,000 replications from samples of size 4,000 generate from DGP 4 (difficult outcome regression / difficult propensity score).
    \end{minipage}
\end{figure}

\section{Empirical Application} \label{sec:empirical_application}

We compare the different approaches in an empirical setting using the dataset from \cite{Knaus:2020}. We analyse the effect of a language course on the employment status of Swiss unemployed in 2003.\footnote{An extensive dataset description can be found in \cite{Knaus:2022}.} The dataset contains 80,706 observations with 2,412 treated and 78,294 control units and includes information on whether an individual participated in a language course ($d = 1$) or not ($d = 0$). The outcome ($Y_i$) measures the number of months an individual was employed during the first six months following the start of the unemployment spell. The ATE is identified in an unconfoundedness setting, assuming that the potential outcomes are independent of the treatment assignment conditional on the covariates. The covariates ($X_i$) include socio-economic characteristics and labour market history. From the literature, we know that in the short run, active labour market programs have a negative effect on the number of months employed due to a lock-in effect \citep[e.g.][]{Lalive:2008,Knaus:2022}.\footnote{Individuals taking part in a course are less likely to search for a job due to time constraints.} Hence, we expect a negative ATE.

The simulation study showed that calibration methods are particularly relevant in small samples, as they tackle a finite sample problem. The empirical application is thus conducted on random sub-samples of the full dataset to investigate the performance of the different calibration methods in small samples. We draw random sub-samples of 12.5\%, 25\%, and 50\% of the treated and the control units separately.\footnote{The random sub-samples of different sizes are drawn incrementally. For example, to get a 25\% sample, we draw the same observations as in the 12.5\% sample and add 12.5\% more observations to the sample. The same is done for the 50\% sample.}

Table \ref{Results_empirical_application_rf} reports the results for a single draw of the sub-samples using a random forest to predict the nuisance functions. The table reports for each calibration method the ATE, its standard error and the Brier score. On the full sample, the ATE is approximately -0.59 across all calibration methods, indicating that visiting a language course leads to a decrease in employment duration of slightly more than half a month in the first six months after the start of the unemployment spell. The result confirms our previous finding, that for large sample sizes the calibration methods do not affect the results. In the following, we will evaluate the results for the different sub-samples in relation to the full sample estimate. Given the large sample size, the standard errors are relatively small and the ATE estimate is likely close to the true quantity of interest. 

A researcher estimating the ATE using DML on a sub-sample of 12.5\% would obtain an estimate of -0.73, concluding that the language course has a 23\% more negative effect compared to the estimate she would obtain using the full sample. After calibrating the propensity scores using the calibration method with the lowest Brier score, the ATE estimate reduces to -0.614, bringing it considerably closer to the full sample estimate. Thus, when working with a small sub-sample, calibrating propensity scores may yield an estimate closer to that of the full sample. As the size of the random sub-samples increases, the DML estimate converges to the full sample estimate, and the differences between the calibrated and uncalibrated DML estimates decrease. This is consistent with theoretical expectations: with larger sample sizes, DML converges to the true effect, the accuracy of propensity score estimation improves, and the need for calibration diminishes.

\begin{table}[h!]
    \centering
    \caption{Results empirical application: Random Forest}
\label{Results_empirical_application_rf}
    \begin{adjustbox}{max width=\textwidth}
        \begin{threeparttable}
    \begin{tabular}{l|ccc|ccc|ccc|ccc} \toprule
        Approach  &\multicolumn{3}{c}{12.5 \% of sample} &\multicolumn{3}{c}{25 \% of sample} &\multicolumn{3}{c}{50 \% of sample} &\multicolumn{3}{c}{full sample}\\  \midrule
        & ATE & SE & Brier score & ATE & SE & Brier score & ATE& SE  & Brier score & ATE & SE & Brier score \\ \midrule
        DML & -0.730 & 0.063 & 0.02878 & -0.678 & 0.042 & 0.02873 & -0.650 & 0.030 & 0.02871 & -0.594 & 0.021 & 0.02871\\
Reweight DML & -0.735 & 0.065 & 0.02878 & -0.677 & 0.046 & 0.02873 & -0.649 & 0.033 & 0.02871 & -0.593 & 0.023 & 0.02871\\
Platt Scaling & -0.614 & 0.089 & \textbf{0.02839} & -0.669 & 0.056 & \textbf{0.02802} & -0.628 & 0.042 & \textbf{0.02797} & -0.587 & 0.030 & 0.02800\\
Beta Scaling & -0.618 & 0.091 & 0.02843 & -0.667 & 0.058 & 0.02806 & -0.628 & 0.043 & 0.02799 & -0.588 & 0.031 & 0.02805\\
Expectation Consistent & -0.638 & 0.065 & 0.02872 & -0.678 & 0.042 & 0.02873 & -0.645 & 0.030 & 0.02870 & -0.594 & 0.021 & 0.02871\\
Temperature Scaling & -0.638 & 0.066 & 0.02873 & -0.678 & 0.044 & 0.02873 & -0.645 & 0.031 & 0.02871 & -0.594 & 0.022 & 0.02872\\
Venn Abers & -0.628 & 0.079 & 0.02870 & -0.675 & 0.053 & 0.02810 & -0.637 & 0.040 & 0.02799 & -0.587 & 0.028 & \textbf{0.02797}\\
 \midrule
    \end{tabular}
    \begin{tablenotes}[flushleft]
     \item \textit{Note:} The table shows the results for the empirical application for different calibration methods and different sample sizes. ATE is the average treatment effect, SE is the standard error of the ATE and the Brier score is defined as $\frac{1}{N}\sum_{i = 1}^N (\hat p(X_i) - D_i)^2$. The random sub-samples of different sizes are drawn incrementally. For example, to get a 25\% sample, we draw the same observations as in the 12.5\% sample and add 12.5\% more observations to the sample.
    \end{tablenotes}
\end{threeparttable}
\end{adjustbox}
\end{table}

The results reported in Table \ref{Results_empirical_application_rf} are confirmed across different random draws of the sub-samples and different machine learners. In Table \ref{Results_empirical_application_simulation} in the Online Appendix \ref{app:additional_results_empirical}, we report average ATE estimates across ten random draws of the sub-samples for random forest, gradient boosting, and Lasso.

\section{Conclusion} \label{sec:conclusion}

This paper investigates the impact of calibration methods on the performance of the double/debiased machine learning estimator in finite samples through a simulation study. We find that calibration methods can significantly improve the performance of the DML estimator when the propensity scores are difficult to estimate and there is high selection into treatment. The improvements are mainly related to a reduction in the bias of the ATE estimate, while the standard deviation remains relatively stable. Venn-Abers calibration, Platt scaling, and Beta scaling perform best among the calibration methods. Isotonic regression performs poorly because some propensity scores are pushed to zero or one, leading to a significant bias. Thus, the choice of the calibration method is crucial and some methods do not work well in the DML framework. However, if a suitable calibration method is chosen, even if there is no need to calibrate the propensity score, it does not hurt to do so. This is validated by our asymptotic result, providing conditions under which the calibrator does not alter the asymptotic properties of the DML estimator. Practitioners should consider the Brier score as a criterion for choosing the calibration method, and should avoid the use of isotonic regression. These results are confirmed in an empirical application, where we find that calibration methods can be particularly relevant in small samples and do not hurt in larger samples.

Future research could investigate the use of calibration methods on the estimation of other treatment effects, such as heterogeneous treatment effects, or in other causal settings, such as difference-in-difference or instrumental variable settings.

\bibliographystyle{apacite}
\bibliography{library}

\begin{thebibliography}{}

\bibitem [\protect \citeauthoryear {%
Ballinari%
}{%
Ballinari%
}{%
{\protect \APACyear {2024}}%
}]{%
Ballinari:2024}
\APACinsertmetastar {%
Ballinari:2024}%
\begin{APACrefauthors}%
Ballinari, D.%
\end{APACrefauthors}%
\unskip\
\newblock
\APACrefYearMonthDay{2024}{}{}.
\newblock
{\BBOQ}\APACrefatitle {Calibrating doubly-robust estimators with unbalanced
  treatment assignment} {Calibrating doubly-robust estimators with unbalanced
  treatment assignment}.{\BBCQ}
\newblock
\APACjournalVolNumPages{Economics Letters}{241}{}{111838}.
\PrintBackRefs{\CurrentBib}

\bibitem [\protect \citeauthoryear {%
Barlow%
\ \BBA {} Brunk%
}{%
Barlow%
\ \BBA {} Brunk%
}{%
{\protect \APACyear {1972}}%
}]{%
Barlow:1972}
\APACinsertmetastar {%
Barlow:1972}%
\begin{APACrefauthors}%
Barlow, R\BPBI E.%
\BCBT {}\ \BBA {} Brunk, H\BPBI D.%
\end{APACrefauthors}%
\unskip\
\newblock
\APACrefYearMonthDay{1972}{}{}.
\newblock
{\BBOQ}\APACrefatitle {The isotonic regression problem and its dual} {The
  isotonic regression problem and its dual}.{\BBCQ}
\newblock
\APACjournalVolNumPages{Journal of the American Statistical
  Association}{67}{337}{140--147}.
\PrintBackRefs{\CurrentBib}

\bibitem [\protect \citeauthoryear {%
Bella%
, Ferri%
, Hern{\'a}ndez-Orallo%
\BCBL {}\ \BBA {} Ram{\'\i}rez-Quintana%
}{%
Bella%
\ \protect \BOthers {.}}{%
{\protect \APACyear {2010}}%
}]{%
Bella:2010}
\APACinsertmetastar {%
Bella:2010}%
\begin{APACrefauthors}%
Bella, A.%
, Ferri, C.%
, Hern{\'a}ndez-Orallo, J.%
\BCBL {}\ \BBA {} Ram{\'\i}rez-Quintana, M\BPBI J.%
\end{APACrefauthors}%
\unskip\
\newblock
\APACrefYearMonthDay{2010}{}{}.
\newblock
{\BBOQ}\APACrefatitle {Calibration of machine learning models} {Calibration of
  machine learning models}.{\BBCQ}
\newblock
\BIn{} \APACrefbtitle {{Handbook of Research on Machine Learning Applications
  and Trends: Algorithms, Methods, and Techniques}} {{Handbook of Research on
  Machine Learning Applications and Trends: Algorithms, Methods, and
  Techniques}}\ (\BPGS\ 128--146).
\newblock
\APACaddressPublisher{}{IGI Global}.
\PrintBackRefs{\CurrentBib}

\bibitem [\protect \citeauthoryear {%
Belloni%
\ \BBA {} Chernozhukov%
}{%
Belloni%
\ \BBA {} Chernozhukov%
}{%
{\protect \APACyear {2013}}%
}]{%
Belloni:2013}
\APACinsertmetastar {%
Belloni:2013}%
\begin{APACrefauthors}%
Belloni, A.%
\BCBT {}\ \BBA {} Chernozhukov, V.%
\end{APACrefauthors}%
\unskip\
\newblock
\APACrefYearMonthDay{2013}{}{}.
\newblock
{\BBOQ}\APACrefatitle {Least squares after model selection in high-dimensional
  sparse models} {Least squares after model selection in high-dimensional
  sparse models}.{\BBCQ}
\newblock
\APACjournalVolNumPages{Bernoulli}{19}{2}{521 -- 547}.
\PrintBackRefs{\CurrentBib}

\bibitem [\protect \citeauthoryear {%
Brier%
}{%
Brier%
}{%
{\protect \APACyear {1950}}%
}]{%
Brier:1950}
\APACinsertmetastar {%
Brier:1950}%
\begin{APACrefauthors}%
Brier, G\BPBI W.%
\end{APACrefauthors}%
\unskip\
\newblock
\APACrefYearMonthDay{1950}{}{}.
\newblock
{\BBOQ}\APACrefatitle {Verification of Forecasts Expressed in Terms of
  Probability} {Verification of forecasts expressed in terms of
  probability}.{\BBCQ}
\newblock
\APACjournalVolNumPages{Monthly Weather Review}{78}{1}{1 -- 3}.
\PrintBackRefs{\CurrentBib}

\bibitem [\protect \citeauthoryear {%
Busso%
, DiNardo%
\BCBL {}\ \BBA {} McCrary%
}{%
Busso%
\ \protect \BOthers {.}}{%
{\protect \APACyear {2014}}%
}]{%
Busso:2014}
\APACinsertmetastar {%
Busso:2014}%
\begin{APACrefauthors}%
Busso, M.%
, DiNardo, J.%
\BCBL {}\ \BBA {} McCrary, J.%
\end{APACrefauthors}%
\unskip\
\newblock
\APACrefYearMonthDay{2014}{}{}.
\newblock
{\BBOQ}\APACrefatitle {New evidence on the finite sample properties of
  propensity score reweighting and matching estimators} {New evidence on the
  finite sample properties of propensity score reweighting and matching
  estimators}.{\BBCQ}
\newblock
\APACjournalVolNumPages{Review of Economics and Statistics}{96}{5}{885--897}.
\PrintBackRefs{\CurrentBib}

\bibitem [\protect \citeauthoryear {%
Caruana%
\ \BBA {} Niculescu-Mizil%
}{%
Caruana%
\ \BBA {} Niculescu-Mizil%
}{%
{\protect \APACyear {2006}}%
}]{%
Caruana:2006}
\APACinsertmetastar {%
Caruana:2006}%
\begin{APACrefauthors}%
Caruana, R.%
\BCBT {}\ \BBA {} Niculescu-Mizil, A.%
\end{APACrefauthors}%
\unskip\
\newblock
\APACrefYearMonthDay{2006}{}{}.
\newblock
{\BBOQ}\APACrefatitle {An empirical comparison of supervised learning
  algorithms} {An empirical comparison of supervised learning
  algorithms}.{\BBCQ}
\newblock
\BIn{} \APACrefbtitle {{Proceedings of the 23rd International Conference on
  Machine Learning}} {{Proceedings of the 23rd International Conference on
  Machine Learning}}\ (\BPGS\ 161--168).
\newblock
\APACaddressPublisher{New York, NY, USA}{Association for Computing Machinery}.
\PrintBackRefs{\CurrentBib}

\bibitem [\protect \citeauthoryear {%
Chernozhukov%
\ \protect \BOthers {.}}{%
Chernozhukov%
\ \protect \BOthers {.}}{%
{\protect \APACyear {2018}}%
}]{%
Chernozhukov:2018}
\APACinsertmetastar {%
Chernozhukov:2018}%
\begin{APACrefauthors}%
Chernozhukov, V.%
, Chetverikov, D.%
, Demirer, M.%
, Duflo, E.%
, Hansen, C.%
, Newey, W.%
\BCBL {}\ \BBA {} Robins, J.%
\end{APACrefauthors}%
\unskip\
\newblock
\APACrefYearMonthDay{2018}{}{}.
\newblock
\APACrefbtitle {Double/debiased machine learning for treatment and structural
  parameters.} {Double/debiased machine learning for treatment and structural
  parameters.}
\newblock
\APACaddressPublisher{}{Oxford University Press Oxford, UK}.
\PrintBackRefs{\CurrentBib}

\bibitem [\protect \citeauthoryear {%
Clart\'e%
, Loureiro%
, Krzakala%
\BCBL {}\ \BBA {} Zdeborov\'a%
}{%
Clart\'e%
\ \protect \BOthers {.}}{%
{\protect \APACyear {2023}}%
}]{%
Clarte:2023}
\APACinsertmetastar {%
Clarte:2023}%
\begin{APACrefauthors}%
Clart\'e, L.%
, Loureiro, B.%
, Krzakala, F.%
\BCBL {}\ \BBA {} Zdeborov\'a, L.%
\end{APACrefauthors}%
\unskip\
\newblock
\APACrefYearMonthDay{2023}{31 Jul--04 Aug}{}.
\newblock
{\BBOQ}\APACrefatitle {Expectation consistency for calibration of neural
  networks} {Expectation consistency for calibration of neural
  networks}.{\BBCQ}
\newblock
\BIn{} R\BPBI J.~Evans\ \BBA {} I.~Shpitser\ (\BEDS), \APACrefbtitle
  {{Proceedings of the Thirty-Ninth Conference on Uncertainty in Artificial
  Intelligence}} {{Proceedings of the Thirty-Ninth Conference on Uncertainty in
  Artificial Intelligence}}\ (\BVOL~216, \BPGS\ 443--453).
\newblock
\APACaddressPublisher{}{PMLR}.
\newblock
\begin{APACrefURL} \url{https://proceedings.mlr.press/v216/clarte23a.html}
  \end{APACrefURL}
\PrintBackRefs{\CurrentBib}

\bibitem [\protect \citeauthoryear {%
Crump%
, Hotz%
, Imbens%
\BCBL {}\ \BBA {} Mitnik%
}{%
Crump%
\ \protect \BOthers {.}}{%
{\protect \APACyear {2009}}%
}]{%
Crump:2009}
\APACinsertmetastar {%
Crump:2009}%
\begin{APACrefauthors}%
Crump, R\BPBI K.%
, Hotz, V\BPBI J.%
, Imbens, G\BPBI W.%
\BCBL {}\ \BBA {} Mitnik, O\BPBI A.%
\end{APACrefauthors}%
\unskip\
\newblock
\APACrefYearMonthDay{2009}{}{}.
\newblock
{\BBOQ}\APACrefatitle {Dealing with limited overlap in estimation of average
  treatment effects} {Dealing with limited overlap in estimation of average
  treatment effects}.{\BBCQ}
\newblock
\APACjournalVolNumPages{Biometrika}{96}{1}{187--199}.
\PrintBackRefs{\CurrentBib}

\bibitem [\protect \citeauthoryear {%
Dal~Pozzolo%
, Caelen%
, Johnson%
\BCBL {}\ \BBA {} Bontempi%
}{%
Dal~Pozzolo%
\ \protect \BOthers {.}}{%
{\protect \APACyear {2015}}%
}]{%
Dal:2015}
\APACinsertmetastar {%
Dal:2015}%
\begin{APACrefauthors}%
Dal~Pozzolo, A.%
, Caelen, O.%
, Johnson, R\BPBI A.%
\BCBL {}\ \BBA {} Bontempi, G.%
\end{APACrefauthors}%
\unskip\
\newblock
\APACrefYearMonthDay{2015}{}{}.
\newblock
{\BBOQ}\APACrefatitle {Calibrating probability with undersampling for
  unbalanced classification} {Calibrating probability with undersampling for
  unbalanced classification}.{\BBCQ}
\newblock
\BIn{} \APACrefbtitle {{2015 IEEE Symposium Series on Computational
  Intelligence}} {{2015 IEEE Symposium Series on Computational Intelligence}}\
  (\BPGS\ 159--166).
\PrintBackRefs{\CurrentBib}

\bibitem [\protect \citeauthoryear {%
Deshpande%
\ \BBA {} Kuleshov%
}{%
Deshpande%
\ \BBA {} Kuleshov%
}{%
{\protect \APACyear {2024}}%
}]{%
Shachi:2024}
\APACinsertmetastar {%
Shachi:2024}%
\begin{APACrefauthors}%
Deshpande, S.%
\BCBT {}\ \BBA {} Kuleshov, V.%
\end{APACrefauthors}%
\unskip\
\newblock
\APACrefYearMonthDay{2024}{}{}.
\newblock
{\BBOQ}\APACrefatitle {Calibrated and Conformal Propensity Scores for Causal
  Effect Estimation} {Calibrated and conformal propensity scores for causal
  effect estimation}.{\BBCQ}
\newblock
\BIn{} \APACrefbtitle {The 40th Conference on Uncertainty in Artificial
  Intelligence.} {The 40th conference on uncertainty in artificial
  intelligence.}
\PrintBackRefs{\CurrentBib}

\bibitem [\protect \citeauthoryear {%
D’Amour%
, Ding%
, Feller%
, Lei%
\BCBL {}\ \BBA {} Sekhon%
}{%
D’Amour%
\ \protect \BOthers {.}}{%
{\protect \APACyear {2021}}%
}]{%
DAmour:2021}
\APACinsertmetastar {%
DAmour:2021}%
\begin{APACrefauthors}%
D’Amour, A.%
, Ding, P.%
, Feller, A.%
, Lei, L.%
\BCBL {}\ \BBA {} Sekhon, J.%
\end{APACrefauthors}%
\unskip\
\newblock
\APACrefYearMonthDay{2021}{}{}.
\newblock
{\BBOQ}\APACrefatitle {Overlap in observational studies with high-dimensional
  covariates} {Overlap in observational studies with high-dimensional
  covariates}.{\BBCQ}
\newblock
\APACjournalVolNumPages{Journal of Econometrics}{221}{2}{644--654}.
\PrintBackRefs{\CurrentBib}

\bibitem [\protect \citeauthoryear {%
Farrell%
, Liang%
\BCBL {}\ \BBA {} Misra%
}{%
Farrell%
\ \protect \BOthers {.}}{%
{\protect \APACyear {2021}}%
}]{%
Farrell:2021}
\APACinsertmetastar {%
Farrell:2021}%
\begin{APACrefauthors}%
Farrell, M\BPBI H.%
, Liang, T.%
\BCBL {}\ \BBA {} Misra, S.%
\end{APACrefauthors}%
\unskip\
\newblock
\APACrefYearMonthDay{2021}{}{}.
\newblock
{\BBOQ}\APACrefatitle {Deep neural networks for estimation and inference} {Deep
  neural networks for estimation and inference}.{\BBCQ}
\newblock
\APACjournalVolNumPages{Econometrica}{89}{1}{181--213}.
\PrintBackRefs{\CurrentBib}

\bibitem [\protect \citeauthoryear {%
Friedman%
}{%
Friedman%
}{%
{\protect \APACyear {1991}}%
}]{%
Friedman:1991}
\APACinsertmetastar {%
Friedman:1991}%
\begin{APACrefauthors}%
Friedman, J\BPBI H.%
\end{APACrefauthors}%
\unskip\
\newblock
\APACrefYearMonthDay{1991}{}{}.
\newblock
{\BBOQ}\APACrefatitle {Multivariate adaptive regression splines} {Multivariate
  adaptive regression splines}.{\BBCQ}
\newblock
\APACjournalVolNumPages{The Annals of Statistics}{19}{1}{1--67}.
\PrintBackRefs{\CurrentBib}

\bibitem [\protect \citeauthoryear {%
Fr{\"o}lich%
}{%
Fr{\"o}lich%
}{%
{\protect \APACyear {2004}}%
}]{%
Froelich:2004}
\APACinsertmetastar {%
Froelich:2004}%
\begin{APACrefauthors}%
Fr{\"o}lich, M.%
\end{APACrefauthors}%
\unskip\
\newblock
\APACrefYearMonthDay{2004}{}{}.
\newblock
{\BBOQ}\APACrefatitle {Finite-sample properties of propensity-score matching
  and weighting estimators} {Finite-sample properties of propensity-score
  matching and weighting estimators}.{\BBCQ}
\newblock
\APACjournalVolNumPages{Review of Economics and Statistics}{86}{1}{77--90}.
\PrintBackRefs{\CurrentBib}

\bibitem [\protect \citeauthoryear {%
Guo%
, Pleiss%
, Sun%
\BCBL {}\ \BBA {} Weinberger%
}{%
Guo%
\ \protect \BOthers {.}}{%
{\protect \APACyear {2017}}%
}]{%
Guo:2017}
\APACinsertmetastar {%
Guo:2017}%
\begin{APACrefauthors}%
Guo, C.%
, Pleiss, G.%
, Sun, Y.%
\BCBL {}\ \BBA {} Weinberger, K\BPBI Q.%
\end{APACrefauthors}%
\unskip\
\newblock
\APACrefYearMonthDay{2017}{}{}.
\newblock
{\BBOQ}\APACrefatitle {On calibration of modern neural networks} {On
  calibration of modern neural networks}.{\BBCQ}
\newblock
\BIn{} \APACrefbtitle {{International Conference on Machine Learning}}
  {{International Conference on Machine Learning}}\ (\BPGS\ 1321--1330).
\PrintBackRefs{\CurrentBib}

\bibitem [\protect \citeauthoryear {%
Gupta%
, Podkopaev%
\BCBL {}\ \BBA {} Ramdas%
}{%
Gupta%
\ \protect \BOthers {.}}{%
{\protect \APACyear {2020}}%
}]{%
Gupta:2020}
\APACinsertmetastar {%
Gupta:2020}%
\begin{APACrefauthors}%
Gupta, C.%
, Podkopaev, A.%
\BCBL {}\ \BBA {} Ramdas, A.%
\end{APACrefauthors}%
\unskip\
\newblock
\APACrefYearMonthDay{2020}{}{}.
\newblock
{\BBOQ}\APACrefatitle {Distribution-free binary classification: prediction
  sets, confidence intervals and calibration} {Distribution-free binary
  classification: prediction sets, confidence intervals and
  calibration}.{\BBCQ}
\newblock
\APACjournalVolNumPages{Advances in Neural Information Processing
  Systems}{33}{}{3711--3723}.
\PrintBackRefs{\CurrentBib}

\bibitem [\protect \citeauthoryear {%
Gupta%
\ \BBA {} Ramdas%
}{%
Gupta%
\ \BBA {} Ramdas%
}{%
{\protect \APACyear {2021}}%
}]{%
Gupta:2021}
\APACinsertmetastar {%
Gupta:2021}%
\begin{APACrefauthors}%
Gupta, C.%
\BCBT {}\ \BBA {} Ramdas, A.%
\end{APACrefauthors}%
\unskip\
\newblock
\APACrefYearMonthDay{2021}{}{}.
\newblock
{\BBOQ}\APACrefatitle {Distribution-free calibration guarantees for histogram
  binning without sample splitting} {Distribution-free calibration guarantees
  for histogram binning without sample splitting}.{\BBCQ}
\newblock
\BIn{} \APACrefbtitle {{International Conference on Machine Learning}}
  {{International Conference on Machine Learning}}\ (\BPGS\ 3942--3952).
\PrintBackRefs{\CurrentBib}

\bibitem [\protect \citeauthoryear {%
Gutman%
, Karavani%
\BCBL {}\ \BBA {} Shimoni%
}{%
Gutman%
\ \protect \BOthers {.}}{%
{\protect \APACyear {2022}}%
}]{%
Gutman:2022}
\APACinsertmetastar {%
Gutman:2022}%
\begin{APACrefauthors}%
Gutman, R.%
, Karavani, E.%
\BCBL {}\ \BBA {} Shimoni, Y.%
\end{APACrefauthors}%
\unskip\
\newblock
\APACrefYearMonthDay{2022}{}{}.
\newblock
{\BBOQ}\APACrefatitle {Propensity score models are better when post-calibrated}
  {Propensity score models are better when post-calibrated}.{\BBCQ}
\newblock
\APACjournalVolNumPages{arXiv preprint arXiv:2211.01221}{}{}{}.
\PrintBackRefs{\CurrentBib}

\bibitem [\protect \citeauthoryear {%
Hahn%
}{%
Hahn%
}{%
{\protect \APACyear {1998}}%
}]{%
Hahn:1998}
\APACinsertmetastar {%
Hahn:1998}%
\begin{APACrefauthors}%
Hahn, J.%
\end{APACrefauthors}%
\unskip\
\newblock
\APACrefYearMonthDay{1998}{}{}.
\newblock
{\BBOQ}\APACrefatitle {On the role of the propensity score in efficient
  semiparametric estimation of average treatment effects} {On the role of the
  propensity score in efficient semiparametric estimation of average treatment
  effects}.{\BBCQ}
\newblock
\APACjournalVolNumPages{Econometrica}{}{}{315--331}.
\PrintBackRefs{\CurrentBib}

\bibitem [\protect \citeauthoryear {%
Huber%
, Lechner%
\BCBL {}\ \BBA {} Wunsch%
}{%
Huber%
\ \protect \BOthers {.}}{%
{\protect \APACyear {2013}}%
}]{%
Huber:2013}
\APACinsertmetastar {%
Huber:2013}%
\begin{APACrefauthors}%
Huber, M.%
, Lechner, M.%
\BCBL {}\ \BBA {} Wunsch, C.%
\end{APACrefauthors}%
\unskip\
\newblock
\APACrefYearMonthDay{2013}{}{}.
\newblock
{\BBOQ}\APACrefatitle {The performance of estimators based on the propensity
  score} {The performance of estimators based on the propensity score}.{\BBCQ}
\newblock
\APACjournalVolNumPages{Journal of Econometrics}{175}{1}{1--21}.
\PrintBackRefs{\CurrentBib}

\bibitem [\protect \citeauthoryear {%
Imbens%
\ \BBA {} Wooldridge%
}{%
Imbens%
\ \BBA {} Wooldridge%
}{%
{\protect \APACyear {2009}}%
}]{%
Imbens:2009}
\APACinsertmetastar {%
Imbens:2009}%
\begin{APACrefauthors}%
Imbens, G\BPBI W.%
\BCBT {}\ \BBA {} Wooldridge, J\BPBI M.%
\end{APACrefauthors}%
\unskip\
\newblock
\APACrefYearMonthDay{2009}{}{}.
\newblock
{\BBOQ}\APACrefatitle {Recent developments in the econometrics of program
  evaluation} {Recent developments in the econometrics of program
  evaluation}.{\BBCQ}
\newblock
\APACjournalVolNumPages{Journal of Economic Literature}{47}{1}{5--86}.
\PrintBackRefs{\CurrentBib}

\bibitem [\protect \citeauthoryear {%
Knaus%
, Lechner%
\BCBL {}\ \BBA {} Strittmatter%
}{%
Knaus%
\ \protect \BOthers {.}}{%
{\protect \APACyear {2021}}%
}]{%
Knaus:2021}
\APACinsertmetastar {%
Knaus:2021}%
\begin{APACrefauthors}%
Knaus, M\BPBI C.%
, Lechner, M.%
\BCBL {}\ \BBA {} Strittmatter, A.%
\end{APACrefauthors}%
\unskip\
\newblock
\APACrefYearMonthDay{2021}{}{}.
\newblock
{\BBOQ}\APACrefatitle {Machine learning estimation of heterogeneous causal
  effects: Empirical Monte Carlo evidence} {Machine learning estimation of
  heterogeneous causal effects: Empirical monte carlo evidence}.{\BBCQ}
\newblock
\APACjournalVolNumPages{The Econometrics Journal}{24}{1}{134--161}.
\PrintBackRefs{\CurrentBib}

\bibitem [\protect \citeauthoryear {%
Knaus%
, Lechner%
\BCBL {}\ \BBA {} Strittmatter%
}{%
Knaus%
\ \protect \BOthers {.}}{%
{\protect \APACyear {2022}}%
}]{%
Knaus:2022}
\APACinsertmetastar {%
Knaus:2022}%
\begin{APACrefauthors}%
Knaus, M\BPBI C.%
, Lechner, M.%
\BCBL {}\ \BBA {} Strittmatter, A.%
\end{APACrefauthors}%
\unskip\
\newblock
\APACrefYearMonthDay{2022}{}{}.
\newblock
{\BBOQ}\APACrefatitle {Heterogeneous employment effects of job search programs:
  A machine learning approach} {Heterogeneous employment effects of job search
  programs: A machine learning approach}.{\BBCQ}
\newblock
\APACjournalVolNumPages{Journal of Human Resources}{57}{2}{597--636}.
\PrintBackRefs{\CurrentBib}

\bibitem [\protect \citeauthoryear {%
Kull%
, Silva~Filho%
\BCBL {}\ \BBA {} Flach%
}{%
Kull%
\ \protect \BOthers {.}}{%
{\protect \APACyear {2017}}%
}]{%
Kull:2017}
\APACinsertmetastar {%
Kull:2017}%
\begin{APACrefauthors}%
Kull, M.%
, Silva~Filho, T\BPBI M.%
\BCBL {}\ \BBA {} Flach, P.%
\end{APACrefauthors}%
\unskip\
\newblock
\APACrefYearMonthDay{2017}{}{}.
\newblock
{\BBOQ}\APACrefatitle {Beyond sigmoids: How to obtain well-calibrated
  probabilities from binary classifiers with beta calibration} {Beyond
  sigmoids: How to obtain well-calibrated probabilities from binary classifiers
  with beta calibration}.{\BBCQ}
\newblock
\APACjournalVolNumPages{Electronic Journal of Statistics}{11}{2}{5052 -- 5080}.
\PrintBackRefs{\CurrentBib}

\bibitem [\protect \citeauthoryear {%
K{\"u}nzel%
, Sekhon%
, Bickel%
\BCBL {}\ \BBA {} Yu%
}{%
K{\"u}nzel%
\ \protect \BOthers {.}}{%
{\protect \APACyear {2019}}%
}]{%
Kunzel:2019}
\APACinsertmetastar {%
Kunzel:2019}%
\begin{APACrefauthors}%
K{\"u}nzel, S\BPBI R.%
, Sekhon, J\BPBI S.%
, Bickel, P\BPBI J.%
\BCBL {}\ \BBA {} Yu, B.%
\end{APACrefauthors}%
\unskip\
\newblock
\APACrefYearMonthDay{2019}{}{}.
\newblock
{\BBOQ}\APACrefatitle {Metalearners for estimating heterogeneous treatment
  effects using machine learning} {Metalearners for estimating heterogeneous
  treatment effects using machine learning}.{\BBCQ}
\newblock
\APACjournalVolNumPages{{Proceedings of the National Academy of
  Sciences}}{116}{10}{4156--4165}.
\PrintBackRefs{\CurrentBib}

\bibitem [\protect \citeauthoryear {%
Lalive%
, Van~Ours%
\BCBL {}\ \BBA {} Zweim{\"u}ller%
}{%
Lalive%
\ \protect \BOthers {.}}{%
{\protect \APACyear {2008}}%
}]{%
Lalive:2008}
\APACinsertmetastar {%
Lalive:2008}%
\begin{APACrefauthors}%
Lalive, R.%
, Van~Ours, J\BPBI C.%
\BCBL {}\ \BBA {} Zweim{\"u}ller, J.%
\end{APACrefauthors}%
\unskip\
\newblock
\APACrefYearMonthDay{2008}{}{}.
\newblock
{\BBOQ}\APACrefatitle {The impact of active labour market programmes on the
  duration of unemployment in Switzerland} {The impact of active labour market
  programmes on the duration of unemployment in switzerland}.{\BBCQ}
\newblock
\APACjournalVolNumPages{The Economic Journal}{118}{525}{235--257}.
\PrintBackRefs{\CurrentBib}

\bibitem [\protect \citeauthoryear {%
Lechner%
\ \protect \BOthers {.}}{%
Lechner%
\ \protect \BOthers {.}}{%
{\protect \APACyear {2020}}%
}]{%
Knaus:2020}
\APACinsertmetastar {%
Knaus:2020}%
\begin{APACrefauthors}%
Lechner, M.%
, Knaus, M\BPBI C.%
, Huber, M.%
, Frölich, M.%
, Behncke, S.%
, Mellace, G.%
\BCBL {}\ \BBA {} Strittmatter, A.%
\end{APACrefauthors}%
\unskip\
\newblock
\APACrefYearMonthDay{2020}{}{}.
\newblock
{\BBOQ}\APACrefatitle {{Swiss Active Labor Market Policy Evaluation [Dataset]}}
  {{Swiss Active Labor Market Policy Evaluation [Dataset]}}.{\BBCQ}
\newblock
\APACjournalVolNumPages{Distributed by FORS, Lausanne}{}{}{}.
\newblock
\begin{APACrefURL} \url{https://doi.org/10.23662/FORS-DS-1203-1}
  \end{APACrefURL}
\PrintBackRefs{\CurrentBib}

\bibitem [\protect \citeauthoryear {%
Lechner%
\ \BBA {} Mareckova%
}{%
Lechner%
\ \BBA {} Mareckova%
}{%
{\protect \APACyear {2024}}%
}]{%
Lechner:2024}
\APACinsertmetastar {%
Lechner:2024}%
\begin{APACrefauthors}%
Lechner, M.%
\BCBT {}\ \BBA {} Mareckova, J.%
\end{APACrefauthors}%
\unskip\
\newblock
\APACrefYearMonthDay{2024}{}{}.
\newblock
{\BBOQ}\APACrefatitle {Comprehensive Causal Machine Learning} {Comprehensive
  causal machine learning}.{\BBCQ}
\newblock
\APACjournalVolNumPages{arXiv preprint arXiv:2405.10198}{}{}{}.
\PrintBackRefs{\CurrentBib}

\bibitem [\protect \citeauthoryear {%
Luo%
, Spindler%
\BCBL {}\ \BBA {} K{\"u}ck%
}{%
Luo%
\ \protect \BOthers {.}}{%
{\protect \APACyear {2016}}%
}]{%
Luo:2016}
\APACinsertmetastar {%
Luo:2016}%
\begin{APACrefauthors}%
Luo, Y.%
, Spindler, M.%
\BCBL {}\ \BBA {} K{\"u}ck, J.%
\end{APACrefauthors}%
\unskip\
\newblock
\APACrefYearMonthDay{2016}{}{}.
\newblock
{\BBOQ}\APACrefatitle {High-Dimensional {$L_2$} Boosting: Rate of Convergence}
  {High-dimensional {$L_2$} boosting: Rate of convergence}.{\BBCQ}
\newblock
\APACjournalVolNumPages{arXiv preprint arXiv:1602.08927}{}{}{}.
\PrintBackRefs{\CurrentBib}

\bibitem [\protect \citeauthoryear {%
Niculescu-Mizil%
\ \BBA {} Caruana%
}{%
Niculescu-Mizil%
\ \BBA {} Caruana%
}{%
{\protect \APACyear {2005}}%
}]{%
Niculescu:2005}
\APACinsertmetastar {%
Niculescu:2005}%
\begin{APACrefauthors}%
Niculescu-Mizil, A.%
\BCBT {}\ \BBA {} Caruana, R.%
\end{APACrefauthors}%
\unskip\
\newblock
\APACrefYearMonthDay{2005}{}{}.
\newblock
{\BBOQ}\APACrefatitle {Predicting good probabilities with supervised learning}
  {Predicting good probabilities with supervised learning}.{\BBCQ}
\newblock
\BIn{} \APACrefbtitle {{Proceedings of the 22nd International Conference on
  Machine Learning}} {{Proceedings of the 22nd International Conference on
  Machine Learning}}\ (\BPGS\ 625--632).
\newblock
\APACaddressPublisher{New York, NY, USA}{Association for Computing Machinery}.
\PrintBackRefs{\CurrentBib}

\bibitem [\protect \citeauthoryear {%
Nie%
\ \BBA {} Wager%
}{%
Nie%
\ \BBA {} Wager%
}{%
{\protect \APACyear {2021}}%
}]{%
Nie:2021}
\APACinsertmetastar {%
Nie:2021}%
\begin{APACrefauthors}%
Nie, X.%
\BCBT {}\ \BBA {} Wager, S.%
\end{APACrefauthors}%
\unskip\
\newblock
\APACrefYearMonthDay{2021}{}{}.
\newblock
{\BBOQ}\APACrefatitle {Quasi-oracle estimation of heterogeneous treatment
  effects} {Quasi-oracle estimation of heterogeneous treatment effects}.{\BBCQ}
\newblock
\APACjournalVolNumPages{Biometrika}{108}{2}{299--319}.
\PrintBackRefs{\CurrentBib}

\bibitem [\protect \citeauthoryear {%
Platt%
}{%
Platt%
}{%
{\protect \APACyear {1999}}%
}]{%
Platt:1999}
\APACinsertmetastar {%
Platt:1999}%
\begin{APACrefauthors}%
Platt, J.%
\end{APACrefauthors}%
\unskip\
\newblock
\APACrefYearMonthDay{1999}{}{}.
\newblock
{\BBOQ}\APACrefatitle {Probabilistic outputs for support vector machines and
  comparisons to regularized likelihood methods} {Probabilistic outputs for
  support vector machines and comparisons to regularized likelihood
  methods}.{\BBCQ}
\newblock
\APACjournalVolNumPages{Advances in Large Margin Classifiers}{10}{3}{61--74}.
\PrintBackRefs{\CurrentBib}

\bibitem [\protect \citeauthoryear {%
Robins%
, Rotnitzky%
\BCBL {}\ \BBA {} Zhao%
}{%
Robins%
\ \protect \BOthers {.}}{%
{\protect \APACyear {1995}}%
}]{%
Robins:1995}
\APACinsertmetastar {%
Robins:1995}%
\begin{APACrefauthors}%
Robins, J\BPBI M.%
, Rotnitzky, A.%
\BCBL {}\ \BBA {} Zhao, L\BPBI P.%
\end{APACrefauthors}%
\unskip\
\newblock
\APACrefYearMonthDay{1995}{}{}.
\newblock
{\BBOQ}\APACrefatitle {Analysis of semiparametric regression models for
  repeated outcomes in the presence of missing data} {Analysis of
  semiparametric regression models for repeated outcomes in the presence of
  missing data}.{\BBCQ}
\newblock
\APACjournalVolNumPages{Journal of the American Statistical
  Association}{90}{429}{106--121}.
\PrintBackRefs{\CurrentBib}

\bibitem [\protect \citeauthoryear {%
Van~der Laan%
, Lin%
, Carone%
\BCBL {}\ \BBA {} Luedtke%
}{%
Van~der Laan%
\ \protect \BOthers {.}}{%
{\protect \APACyear {2024}}%
}]{%
VanderLaan:2024a}
\APACinsertmetastar {%
VanderLaan:2024a}%
\begin{APACrefauthors}%
Van~der Laan, L.%
, Lin, Z.%
, Carone, M.%
\BCBL {}\ \BBA {} Luedtke, A.%
\end{APACrefauthors}%
\unskip\
\newblock
\APACrefYearMonthDay{2024}{}{}.
\newblock
{\BBOQ}\APACrefatitle {Stabilized Inverse Probability Weighting via Isotonic
  Calibration} {Stabilized inverse probability weighting via isotonic
  calibration}.{\BBCQ}
\newblock
\APACjournalVolNumPages{arXiv preprint arXiv:2411.06342}{}{}{}.
\PrintBackRefs{\CurrentBib}

\bibitem [\protect \citeauthoryear {%
Van~der Laan%
, Ulloa-P{\'e}rez%
, Carone%
\BCBL {}\ \BBA {} Luedtke%
}{%
Van~der Laan%
\ \protect \BOthers {.}}{%
{\protect \APACyear {2023}}%
}]{%
vanderLaan:2023}
\APACinsertmetastar {%
vanderLaan:2023}%
\begin{APACrefauthors}%
Van~der Laan, L.%
, Ulloa-P{\'e}rez, E.%
, Carone, M.%
\BCBL {}\ \BBA {} Luedtke, A.%
\end{APACrefauthors}%
\unskip\
\newblock
\APACrefYearMonthDay{2023}{}{}.
\newblock
{\BBOQ}\APACrefatitle {Causal isotonic calibration for heterogeneous treatment
  effects} {Causal isotonic calibration for heterogeneous treatment
  effects}.{\BBCQ}
\newblock
\APACjournalVolNumPages{arXiv preprint arXiv:2302.14011}{}{}{}.
\PrintBackRefs{\CurrentBib}

\bibitem [\protect \citeauthoryear {%
Vovk%
\ \BBA {} Petej%
}{%
Vovk%
\ \BBA {} Petej%
}{%
{\protect \APACyear {2014}}%
}]{%
Vovk:2012}
\APACinsertmetastar {%
Vovk:2012}%
\begin{APACrefauthors}%
Vovk, V.%
\BCBT {}\ \BBA {} Petej, I.%
\end{APACrefauthors}%
\unskip\
\newblock
\APACrefYearMonthDay{2014}{}{}.
\newblock
{\BBOQ}\APACrefatitle {Venn-abers predictors} {Venn-abers predictors}.{\BBCQ}
\newblock
\APACjournalVolNumPages{arXiv preprint arXiv:1211.0025}{}{}{}.
\PrintBackRefs{\CurrentBib}

\bibitem [\protect \citeauthoryear {%
Vovk%
, Petej%
\BCBL {}\ \BBA {} Fedorova%
}{%
Vovk%
\ \protect \BOthers {.}}{%
{\protect \APACyear {2015}}%
}]{%
Vovk:2015}
\APACinsertmetastar {%
Vovk:2015}%
\begin{APACrefauthors}%
Vovk, V.%
, Petej, I.%
\BCBL {}\ \BBA {} Fedorova, V.%
\end{APACrefauthors}%
\unskip\
\newblock
\APACrefYearMonthDay{2015}{}{}.
\newblock
{\BBOQ}\APACrefatitle {Large-scale probabilistic prediction with and without
  validity guarantees} {Large-scale probabilistic prediction with and without
  validity guarantees}.{\BBCQ}
\newblock
\BIn{} \APACrefbtitle {{Proceedings of NIPS}} {{Proceedings of NIPS}}\ (\BVOL\
  2015).
\PrintBackRefs{\CurrentBib}

\bibitem [\protect \citeauthoryear {%
Wager%
\ \BBA {} Walther%
}{%
Wager%
\ \BBA {} Walther%
}{%
{\protect \APACyear {2015}}%
}]{%
Wager:2015}
\APACinsertmetastar {%
Wager:2015}%
\begin{APACrefauthors}%
Wager, S.%
\BCBT {}\ \BBA {} Walther, G.%
\end{APACrefauthors}%
\unskip\
\newblock
\APACrefYearMonthDay{2015}{}{}.
\newblock
{\BBOQ}\APACrefatitle {Adaptive concentration of regression trees, with
  application to random forests} {Adaptive concentration of regression trees,
  with application to random forests}.{\BBCQ}
\newblock
\APACjournalVolNumPages{arXiv preprint arXiv:1503.06388}{}{}{}.
\PrintBackRefs{\CurrentBib}

\bibitem [\protect \citeauthoryear {%
Wang%
}{%
Wang%
}{%
{\protect \APACyear {2023}}%
}]{%
Wang:2023}
\APACinsertmetastar {%
Wang:2023}%
\begin{APACrefauthors}%
Wang, C.%
\end{APACrefauthors}%
\unskip\
\newblock
\APACrefYearMonthDay{2023}{}{}.
\newblock
{\BBOQ}\APACrefatitle {Calibration in deep learning: A survey of the
  state-of-the-art} {Calibration in deep learning: A survey of the
  state-of-the-art}.{\BBCQ}
\newblock
\APACjournalVolNumPages{arXiv preprint arXiv:2308.01222}{}{}{}.
\PrintBackRefs{\CurrentBib}

\bibitem [\protect \citeauthoryear {%
Xu%
\ \BBA {} Yadlowsky%
}{%
Xu%
\ \BBA {} Yadlowsky%
}{%
{\protect \APACyear {2022}}%
}]{%
Xu:2022}
\APACinsertmetastar {%
Xu:2022}%
\begin{APACrefauthors}%
Xu, Y.%
\BCBT {}\ \BBA {} Yadlowsky, S.%
\end{APACrefauthors}%
\unskip\
\newblock
\APACrefYearMonthDay{2022}{}{}.
\newblock
{\BBOQ}\APACrefatitle {Calibration error for heterogeneous treatment effects}
  {Calibration error for heterogeneous treatment effects}.{\BBCQ}
\newblock
\BIn{} \APACrefbtitle {{International Conference on Artificial Intelligence and
  Statistics}} {{International Conference on Artificial Intelligence and
  Statistics}}\ (\BPGS\ 9280--9303).
\PrintBackRefs{\CurrentBib}

\bibitem [\protect \citeauthoryear {%
Yang%
\ \BBA {} Ding%
}{%
Yang%
\ \BBA {} Ding%
}{%
{\protect \APACyear {2018}}%
}]{%
Yang:2018}
\APACinsertmetastar {%
Yang:2018}%
\begin{APACrefauthors}%
Yang, S.%
\BCBT {}\ \BBA {} Ding, P.%
\end{APACrefauthors}%
\unskip\
\newblock
\APACrefYearMonthDay{2018}{03}{}.
\newblock
{\BBOQ}\APACrefatitle {{Asymptotic inference of causal effects with
  observational studies trimmed by the estimated propensity scores}}
  {{Asymptotic inference of causal effects with observational studies trimmed
  by the estimated propensity scores}}.{\BBCQ}
\newblock
\APACjournalVolNumPages{Biometrika}{105}{2}{487--493}.
\PrintBackRefs{\CurrentBib}

\bibitem [\protect \citeauthoryear {%
Zadrozny%
\ \BBA {} Elkan%
}{%
Zadrozny%
\ \BBA {} Elkan%
}{%
{\protect \APACyear {2002}}%
}]{%
Zadrozny:2002}
\APACinsertmetastar {%
Zadrozny:2002}%
\begin{APACrefauthors}%
Zadrozny, B.%
\BCBT {}\ \BBA {} Elkan, C.%
\end{APACrefauthors}%
\unskip\
\newblock
\APACrefYearMonthDay{2002}{}{}.
\newblock
{\BBOQ}\APACrefatitle {Transforming classifier scores into accurate multiclass
  probability estimates} {Transforming classifier scores into accurate
  multiclass probability estimates}.{\BBCQ}
\newblock
\BIn{} \APACrefbtitle {{Proceedings of the eighth ACM SIGKDD International
  Conference on Knowledge Discovery and Data Mining}} {{Proceedings of the
  eighth ACM SIGKDD International Conference on Knowledge Discovery and Data
  Mining}}\ (\BPGS\ 694--699).
\PrintBackRefs{\CurrentBib}

\bibitem [\protect \citeauthoryear {%
Zhang%
}{%
Zhang%
}{%
{\protect \APACyear {2002}}%
}]{%
Zhang:2002}
\APACinsertmetastar {%
Zhang:2002}%
\begin{APACrefauthors}%
Zhang, C\BHBI H.%
\end{APACrefauthors}%
\unskip\
\newblock
\APACrefYearMonthDay{2002}{}{}.
\newblock
{\BBOQ}\APACrefatitle {Risk Bounds in Isotonic Regression} {Risk bounds in
  isotonic regression}.{\BBCQ}
\newblock
\APACjournalVolNumPages{The Annals of Statistics}{30}{2}{528--555}.
\PrintBackRefs{\CurrentBib}

\end{thebibliography}
\clearpage
\begin{appendices}

\section{Asymptotic Proofs} \label{app:asymptotic_proofs}
\renewcommand{\theequation}{\thesection\arabic{equation}}
\setcounter{equation}{0}
\renewcommand{\thelemma}{\thesection\arabic{lemma}}
\setcounter{lemma}{0}
\renewcommand{\thetable}{\thesection\arabic{table}}
\setcounter{table}{0}

\begin{proof}[Proof of Theorem \ref{theorem:convergence}]
    We will show that $\tilde{f}(\hat{p}(x))$ converges to $p(x)$. Let $\mathcal{E}_N$ denote the event $\mu_d^*, p^* \in \Xi_N$ and $f^* \in \Lambda_{N, p^*}$. Note that, conditionally on the training and calibration samples $\mathcal{S}_T$ and $\mathcal{S}_C$, the nuisance functions are nonstochastic. Than, on the event $\mathcal{E}_N$:
    \begin{align*}
        \E\Big[\big(\tilde{f}(\hat{p}(X_i)) - p(X_i)\big)^2 \Big\vert \mathcal{S}_T, \mathcal{S}_C \Big]^{1/2} &\leq \sup_{p^* \in \Xi_N} \sup_{f^*\in\Lambda_{N,p^*}} \E\Big[\big(f^*(p^*(X_i)) - p(X_i)\big)^2 \Big\vert \mathcal{S}_T, \mathcal{S}_C \Big]^{1/2}\\
        &\leq\sup_{p^* \in \Xi_N} \sup_{f^*\in\Lambda_{N,p^*}} \norm{f^*(p^*(X_i)) - p(X_i)}_2\\
        &\leq \sup_{p^* \in \Xi_N} \sup_{f^*\in\Lambda_{N,p^*}} \norm{f^*(p^*(X_i)) - f(p^*(X_i)) + f(p^*(X_i)) - p(X_i)}_2\\
        &\leq \underbrace{\sup_{p^* \in \Xi_N} \sup_{f^*\in\Lambda_{N,p^*}} \norm{f^*(p^*(X_i)) - f(p^*(X_i))}_2}_{\text{Part I}}\\ &\phantom{\leq}+ \underbrace{\sup_{p^* \in \Xi_N} \norm{f(p^*(X_i)) - p(X_i)}_2}_{\text{Part II}}
    \end{align*}
    where the last equality follows from Minkowski's inequality. Part I corresponds to $r_{f,N}$ by its definition in Assumption \ref{ass:calibrator}. Part II is bounded by $r_{p,N}$ by Assumption \ref{ass:calibrator}-(v). By Lemma 6.1 in \cite{Chernozhukov:2018} and by Assumption \ref{ass:dml} we have $\smallnorm{\tilde{f}(\hat{p}(X_i)) - p(X_i)}_2=o_p(r_{f,N} + r_{p,N})$. In particular, we have that $\smallnorm{\tilde{f}(\hat{p}(X_i)) - p(X_i)}_2 \smallnorm{\hat{\mu}_d(X_i) - \mu_d(X_i)}_2 = o_p(N^{-1/2})$. This concludes the proof as from here it follows the same arguments as the proof of Theorem 5.1 in \cite{Chernozhukov:2018}.
\end{proof}

\section{Hyperparameter Tuning for Different Estimators}\label{app:grid_search}
\renewcommand{\theequation}{\thesection\arabic{equation}}
\setcounter{equation}{0}
\renewcommand{\thelemma}{\thesection\arabic{lemma}}
\setcounter{lemma}{0}
\renewcommand{\thetable}{\thesection\arabic{table}}
\setcounter{table}{0}

\begin{table}[h!]
    \centering
    \caption{Grid search for tuning parameters of different estimators}
\label{grid_search}
\begin{adjustbox}{max width=0.6\textwidth}
    \begin{threeparttable}
\begin{tabular}{ll} \toprule
    \multicolumn{2}{l}{\textbf{Random Forest}} \\ \midrule
    Maximum depth & 1, 2, 3, 5, 10, 20 \\
    Minimum leaf size & 5, 10, 15, 20, 30, 50 \\ \midrule
    \multicolumn{2}{l}{\textbf{Gradient Boosting}} \\ \midrule
    Number of estimators & 5, 10, 25, 50, 100, 200, 500 \\
    Learning rate & 0.001, 0.005, 0.01, 0.05, 0.1, 0.5, 1 \\
    Maximum depth & 1, 2, 3, 5, 10 \\ \midrule
    \multicolumn{2}{l}{\textbf{Lasso}} \\ \midrule
    Regularization parameter & 0.005, 0.01, 0.05, 0.1, 0.5, 0.8, 1 \\ \midrule
\end{tabular}
\begin{tablenotes}[flushleft]
    \small
 \item \textit{Note:} The table shows grid search for the tuning parameters of the different estimators.
\end{tablenotes}
\end{threeparttable}
\end{adjustbox}
\end{table}

\begin{table}[h!]
    \centering
    \caption{Selected hyperparameters for the different DGPs and sample sizes}
\label{selected hyperparameters}
    \begin{adjustbox}{max width=\textwidth}
        \begin{threeparttable}
    \begin{tabular}{ll|ccc|ccc|ccc|ccc|ccc|ccc} \toprule
        
        &  &\multicolumn{6}{c}{Random Forest} &\multicolumn{9}{c}{Gradient Boosting} &\multicolumn{3}{c}{Lasso} \\  \cmidrule(lr){3-8} \cmidrule(lr){9-17} \cmidrule(lr){18-20}
          &  &\multicolumn{3}{c}{max depth} &\multicolumn{3}{c}{min leaf size} &\multicolumn{3}{c}{max depth} &\multicolumn{3}{c}{learning rate} &\multicolumn{3}{c}{nr. estimators} &\multicolumn{3}{c}{regularization strength} \\  \cmidrule(lr){3-8} \cmidrule(lr){9-17} \cmidrule(lr){18-20}
        N & DGP & $p(x)$ & $\mu(1,x)$ & $\mu(0,x)$ & $p(x)$ & $\mu(1,x)$ & $\mu(0,x)$ & $p(x)$ & $\mu(1,x)$ & $\mu(0,x)$ & $p(x)$ & $\mu(1,x)$ & $\mu(0,x)$ & $p(x)$ & $\mu(1,x)$ & $\mu(0,x)$ & $p(x)$ & $\mu(1,x)$ & $\mu(0,x)$ \\ \midrule
        2000 & 1 & 10 & 10 & 3 & 20 & 20 & 30 & 1 & 1 & 1 & 0.1 & 0.1 & 0.01 & 100 & 100 & 200 & 0.1 & 0.1 & 0.1\\
2000 & 2 & 5 & 5 & 2 & 5 & 5 & 30 & 1 & 1 & 1 & 0.1 & 0.1 & 0.5 & 100 & 100 & 500 & 0.1 & 0.1 & 0.1\\
2000 & 3 & 5 & 5 & 5 & 30 & 30 & 15 & 1 & 1 & 1 & 0.05 & 0.05 & 0.1 & 200 & 200 & 50 & 0.1 & 0.1 & 0.1\\
2000 & 4 & 5 & 5 & 10 & 5 & 5 & 30 & 1 & 1 & 1 & 0.1 & 0.1 & 0.05 & 50 & 50 & 100 & 0.1 & 0.1 & 0.1\\
2000 & 5 & 5 & 5 & 2 & 30 & 30 & 50 & 1 & 1 & 1 & 0.05 & 0.05 & 0.05 & 200 & 200 & 50 & 0.1 & 0.1 & 0.1\\
2000 & 6 & 3 & 3 & 20 & 15 & 15 & 30 & 1 & 1 & 1 & 0.05 & 0.05 & 0.1 & 100 & 100 & 100 & 0.1 & 0.1 & 0.05\\
4000 & 1 & 10 & 10 & 5 & 30 & 30 & 50 & 1 & 1 & 1 & 0.1 & 0.1 & 0.1 & 100 & 100 & 200 & 0.05 & 0.05 & 0.05\\
4000 & 2 & 10 & 10 & 5 & 20 & 20 & 50 & 1 & 1 & 1 & 0.1 & 0.1 & 0.1 & 100 & 100 & 100 & 0.05 & 0.05 & 0.05\\
4000 & 3 & 20 & 20 & 10 & 30 & 30 & 50 & 1 & 1 & 1 & 0.1 & 0.1 & 0.1 & 100 & 100 & 200 & 0.05 & 0.05 & 0.05\\
4000 & 4 & 5 & 5 & 10 & 30 & 30 & 50 & 1 & 1 & 1 & 0.1 & 0.1 & 0.05 & 100 & 100 & 200 & 0.05 & 0.05 & 0.05\\
4000 & 5 & 20 & 20 & 10 & 30 & 30 & 50 & 1 & 1 & 1 & 0.1 & 0.1 & 0.1 & 200 & 200 & 100 & 0.05 & 0.05 & 0.05\\
4000 & 6 & 5 & 5 & 10 & 30 & 30 & 50 & 1 & 1 & 1 & 0.1 & 0.1 & 0.1 & 100 & 100 & 100 & 0.05 & 0.05 & 0.05\\
8000 & 1 & 20 & 20 & 5 & 30 & 30 & 50 & 1 & 1 & 1 & 0.1 & 0.1 & 0.1 & 200 & 200 & 100 & 0.05 & 0.05 & 0.05\\
8000 & 2 & 10 & 10 & 5 & 30 & 30 & 50 & 1 & 1 & 1 & 0.1 & 0.1 & 0.1 & 200 & 200 & 100 & 0.05 & 0.05 & 0.01\\
8000 & 3 & 10 & 10 & 20 & 30 & 30 & 50 & 2 & 2 & 1 & 0.1 & 0.1 & 0.1 & 200 & 200 & 200 & 0.01 & 0.01 & 0.05\\
8000 & 4 & 5 & 5 & 20 & 30 & 30 & 50 & 1 & 1 & 2 & 0.1 & 0.1 & 0.1 & 100 & 100 & 200 & 0.05 & 0.05 & 0.01\\
8000 & 5 & 20 & 20 & 5 & 50 & 50 & 50 & 1 & 1 & 1 & 0.1 & 0.1 & 0.05 & 200 & 200 & 200 & 0.05 & 0.05 & 0.05\\
8000 & 6 & 10 & 10 & 20 & 20 & 20 & 50 & 1 & 1 & 2 & 0.1 & 0.1 & 0.1 & 200 & 200 & 200 & 0.01 & 0.01 & 0.05\\

 \midrule
    \end{tabular}
    \begin{tablenotes}[flushleft]
     \item \textit{Note:} The table shows the different hyperparameters that have been selected for the different DGPs and sample sizes. Column (3) to (8) show them for the Random Forest, columns (9) to (17) for the Gradient Boosting and columns (18) to (20) for the Lasso.
    \end{tablenotes}
\end{threeparttable}
\end{adjustbox}
\end{table}
\clearpage

\section{Online Appendix} \label{app:online_appendix}
\renewcommand{\theequation}{\thesection\arabic{equation}}
\setcounter{equation}{0}
\renewcommand{\thelemma}{\thesection\arabic{lemma}}
\setcounter{lemma}{0}
\renewcommand{\thetable}{\thesection\arabic{table}}
\setcounter{table}{0}

\subsection{Additional Simulation Results} \label{app:additional_results}

\begin{table}[h!]
    \centering
    \caption{Simulation results for sample size $N = 4,000$}
\label{Results_simulation_4000}
    \begin{adjustbox}{max width=\textwidth}
        \begin{threeparttable}
    \begin{tabular}{l|cccc|cccc|cccc} \toprule
        
        Approach  &\multicolumn{4}{c}{Random Forest} &\multicolumn{4}{c}{Gradient Boosting} &\multicolumn{4}{c}{Lasso} \\  \midrule
        & RMSE & Bias & St. Dev. & Coverage & RMSE & Bias & St. Dev. & Coverage & RMSE & Bias & St. Dev. & Coverage \\ \midrule
        &\multicolumn{12}{c}{DGP 1 (easy outcome regression / easy propensity scores)} \\  \midrule
        DML & 0.036 & -0.012 & 0.034 & 0.939 & 0.036 & -0.014 & 0.034 & 0.925 & 0.035 & -0.009 & 0.034 & 0.948\\
Reweight DML & 0.036 & -0.012 & 0.034 & 0.939 & 0.036 & -0.014 & 0.034 & 0.925 & 0.035 & -0.009 & 0.034 & 0.948\\
Platt Scaling & 0.034 & -0.006 & 0.034 & 0.949 & 0.035 & -0.008 & 0.034 & 0.946 & 0.034 & -0.005 & 0.034 & 0.949\\
Beta Scaling & \textbf{0.034} & -0.006 & 0.034 & 0.949 & 0.035 & -0.008 & 0.034 & 0.945 & \textbf{0.034} & -0.005 & 0.034 & 0.950\\
Expectation Consistent & 0.035 & -0.007 & 0.034 & 0.950 & 0.035 & -0.009 & 0.034 & 0.942 & 0.034 & -0.005 & 0.034 & 0.952\\
Temperature Scaling & 0.035 & -0.007 & 0.034 & 0.950 & 0.035 & -0.009 & 0.034 & 0.946 & 0.034 & -0.005 & 0.034 & 0.951\\
Isotonic Regression & 0.608 & 0.095 & 0.600 & 0.942 & \textbf{0.034} & -0.007 & 0.034 & 0.941 & 0.630 & 0.091 & 0.623 & 0.941\\
Venn Abers & 0.034 & -0.005 & 0.034 & 0.944 & 0.034 & -0.007 & 0.034 & 0.941 & 0.034 & -0.004 & 0.034 & 0.947\\
 \midrule
        &\multicolumn{12}{c}{DGP 2 (easy outcome regression / difficult propensity scores)} \\  \midrule
        DML & 0.054 & 0.038 & 0.038 & 0.823 & 0.062 & 0.049 & 0.038 & 0.734 & 0.045 & 0.025 & 0.037 & 0.893\\
Reweight DML & 0.053 & 0.036 & 0.038 & 0.833 & 0.060 & 0.047 & 0.038 & 0.749 & 0.044 & 0.024 & 0.037 & 0.894\\
Platt Scaling & \textbf{0.041} & 0.012 & 0.039 & 0.941 & 0.041 & 0.010 & 0.040 & 0.946 & 0.039 & -0.007 & 0.038 & 0.952\\
Beta Scaling & 0.041 & 0.010 & 0.040 & 0.946 & 0.041 & 0.009 & 0.039 & 0.950 & \textbf{0.039} & -0.007 & 0.038 & 0.956\\
Expectation Consistent & 0.042 & 0.012 & 0.041 & 0.941 & 0.043 & 0.006 & 0.043 & 0.959 & 0.062 & -0.042 & 0.045 & 0.844\\
Temperature Scaling & 0.043 & 0.017 & 0.039 & 0.930 & 0.045 & 0.023 & 0.039 & 0.907 & 0.047 & -0.024 & 0.040 & 0.924\\
Isotonic Regression & 0.701 & -0.306 & 0.631 & 0.918 & 0.399 & -0.064 & 0.394 & 0.932 & 0.715 & -0.255 & 0.668 & 0.929\\
Venn Abers & 0.041 & 0.010 & 0.040 & 0.949 & \textbf{0.040} & 0.004 & 0.040 & 0.953 & 0.040 & -0.003 & 0.040 & 0.955\\
 \midrule
        &\multicolumn{12}{c}{DGP 3 (difficult outcome regression / easy propensity scores)} \\  \midrule
        DML & 0.036 & -0.011 & 0.034 & 0.933 & 0.037 & -0.014 & 0.034 & 0.929 & 0.036 & -0.012 & 0.034 & 0.938\\
Reweight DML & 0.036 & -0.011 & 0.034 & 0.933 & 0.037 & -0.014 & 0.034 & 0.929 & 0.036 & -0.012 & 0.034 & 0.938\\
Platt Scaling & \textbf{0.035} & -0.006 & 0.034 & 0.950 & 0.035 & -0.007 & 0.034 & 0.951 & 0.035 & -0.007 & 0.035 & 0.949\\
Beta Scaling & 0.035 & -0.006 & 0.034 & 0.948 & 0.035 & -0.007 & 0.034 & 0.951 & \textbf{0.035} & -0.007 & 0.035 & 0.949\\
Expectation Consistent & 0.035 & -0.007 & 0.034 & 0.943 & 0.035 & -0.009 & 0.034 & 0.945 & 0.035 & -0.007 & 0.035 & 0.949\\
Temperature Scaling & 0.035 & -0.007 & 0.034 & 0.945 & 0.035 & -0.009 & 0.034 & 0.946 & 0.035 & -0.007 & 0.035 & 0.949\\
Isotonic Regression & 0.611 & 0.102 & 0.603 & 0.941 & \textbf{0.035} & -0.006 & 0.034 & 0.948 & 0.633 & 0.094 & 0.626 & 0.945\\
Venn Abers & 0.035 & -0.004 & 0.034 & 0.949 & 0.035 & -0.006 & 0.034 & 0.948 & 0.035 & -0.005 & 0.035 & 0.942\\
 \midrule
        &\multicolumn{12}{c}{DGP 4 (difficult outcome regression / difficult propensity scores)} \\  \midrule
        DML & 0.059 & 0.043 & 0.039 & 0.791 & 0.065 & 0.053 & 0.039 & 0.721 & 0.110 & 0.103 & 0.038 & 0.226\\
Reweight DML & 0.057 & 0.041 & 0.039 & 0.814 & 0.063 & 0.050 & 0.039 & 0.744 & 0.107 & 0.100 & 0.038 & 0.256\\
Platt Scaling & 0.042 & 0.013 & 0.040 & 0.938 & 0.040 & 0.005 & 0.040 & 0.955 & 0.089 & 0.080 & 0.039 & 0.479\\
Beta Scaling & 0.041 & 0.009 & 0.040 & 0.949 & \textbf{0.040} & 0.004 & 0.040 & 0.954 & 0.072 & 0.060 & 0.039 & 0.656\\
Expectation Consistent & 0.043 & 0.009 & 0.042 & 0.944 & 0.044 & 0.001 & 0.044 & 0.957 & 0.143 & 0.135 & 0.049 & 0.208\\
Temperature Scaling & 0.043 & 0.016 & 0.040 & 0.933 & 0.045 & 0.021 & 0.040 & 0.916 & 0.124 & 0.116 & 0.042 & 0.204\\
Isotonic Regression & 0.705 & -0.291 & 0.642 & 0.919 & 0.406 & -0.073 & 0.400 & 0.933 & 0.881 & 0.338 & 0.814 & 0.932\\
Venn Abers & \textbf{0.041} & 0.009 & 0.040 & 0.948 & 0.041 & -0.000 & 0.041 & 0.955 & \textbf{0.042} & 0.011 & 0.040 & 0.944\\
 \midrule
        &\multicolumn{12}{c}{DGP 5 (easy outcome regression / difficult and extreme propensity scores)} \\  \midrule
        DML & 0.090 & 0.081 & 0.038 & 0.444 & 0.051 & 0.024 & 0.046 & 0.918 & 0.101 & -0.054 & 0.085 & 0.910\\
Reweight DML & 0.085 & 0.075 & 0.039 & 0.500 & 0.051 & 0.022 & 0.046 & 0.921 & 0.082 & -0.038 & 0.072 & 0.920\\
Platt Scaling & 0.052 & 0.028 & 0.043 & 0.895 & \textbf{0.049} & 0.020 & 0.044 & 0.925 & \textbf{0.045} & 0.004 & 0.045 & 0.954\\
Beta Scaling & 0.051 & 0.023 & 0.046 & 0.913 & 0.049 & 0.012 & 0.048 & 0.942 & 0.057 & -0.008 & 0.056 & 0.949\\
Expectation Consistent & 0.053 & 0.028 & 0.044 & 0.887 & 0.051 & 0.018 & 0.047 & 0.926 & 0.114 & -0.060 & 0.097 & 0.920\\
Temperature Scaling & \textbf{0.051} & 0.025 & 0.045 & 0.906 & 0.050 & 0.016 & 0.047 & 0.930 & 0.083 & -0.038 & 0.073 & 0.930\\
Isotonic Regression & 0.794 & -0.396 & 0.688 & 0.911 & 0.779 & -0.298 & 0.719 & 0.933 & 0.824 & -0.397 & 0.722 & 0.914\\
Venn Abers & 0.052 & 0.024 & 0.046 & 0.921 & 0.050 & 0.012 & 0.049 & 0.945 & 0.049 & 0.002 & 0.049 & 0.947\\
 \midrule
        &\multicolumn{12}{c}{DGP 6 (difficult outcome regression / difficult and extreme propensity scores)} \\  \midrule
        DML & 0.099 & 0.090 & 0.042 & 0.436 & 0.047 & 0.007 & 0.046 & 0.951 & 0.199 & 0.166 & 0.109 & 0.703\\
Reweight DML & 0.091 & 0.081 & 0.042 & 0.518 & 0.047 & 0.004 & 0.046 & 0.954 & 0.167 & 0.142 & 0.090 & 0.685\\
Platt Scaling & \textbf{0.046} & 0.009 & 0.045 & 0.955 & \textbf{0.045} & -0.001 & 0.045 & 0.954 & \textbf{0.050} & 0.018 & 0.047 & 0.931\\
Beta Scaling & 0.047 & 0.005 & 0.047 & 0.954 & 0.049 & -0.009 & 0.048 & 0.945 & 0.068 & 0.022 & 0.064 & 0.936\\
Expectation Consistent & 0.047 & 0.011 & 0.046 & 0.949 & 0.048 & -0.002 & 0.048 & 0.947 & 0.218 & 0.177 & 0.128 & 0.768\\
Temperature Scaling & 0.046 & 0.005 & 0.046 & 0.956 & 0.048 & -0.007 & 0.048 & 0.944 & 0.170 & 0.145 & 0.089 & 0.664\\
Isotonic Regression & 0.846 & -0.469 & 0.704 & 0.904 & 0.829 & -0.387 & 0.733 & 0.916 & 0.856 & 0.268 & 0.813 & 0.935\\
Venn Abers & 0.048 & 0.006 & 0.047 & 0.950 & 0.049 & -0.005 & 0.049 & 0.949 & 0.054 & 0.016 & 0.051 & 0.936\\
 \midrule
    \end{tabular}
    \begin{tablenotes}[flushleft]
     \item \textit{Note:} The table shows the results for the different calibration methods for the different DGPs with $N = 4,000$. The RMSE is the root mean squared error, the bias is the average difference between the true value and the estimated value, the standard deviation is the standard deviation of the estimated values and the coverage is the coverage of the 95\% confidence interval.
    \end{tablenotes}
\end{threeparttable}
\end{adjustbox}
\end{table}

\begin{table}[h!]
    \centering
    \caption{Simulation results for sample size $N = 8,000$}
\label{Results_simulation_8000}
    \begin{adjustbox}{max width=\textwidth}
        \begin{threeparttable}
    \begin{tabular}{l|cccc|cccc|cccc} \toprule
        
        Approach  &\multicolumn{4}{c}{Random Forest} &\multicolumn{4}{c}{Gradient Boosting} &\multicolumn{4}{c}{Lasso} \\  \midrule
        & RMSE & Bias & St. Dev. & Coverage & RMSE & Bias & St. Dev. & Coverage & RMSE & Bias & St. Dev. & Coverage \\ \midrule
        &\multicolumn{12}{c}{DGP 1 (easy outcome regression / easy propensity scores)} \\  \midrule
        DML & 0.024 & -0.005 & 0.024 & 0.945 & \textbf{0.024} & -0.002 & 0.024 & 0.948 & 0.024 & -0.005 & 0.024 & 0.945\\
Reweight DML & 0.024 & -0.005 & 0.024 & 0.945 & 0.024 & -0.003 & 0.024 & 0.948 & 0.024 & -0.005 & 0.024 & 0.945\\
Platt Scaling & 0.024 & -0.004 & 0.024 & 0.945 & 0.024 & -0.004 & 0.024 & 0.943 & 0.024 & -0.003 & 0.024 & 0.948\\
Beta Scaling & 0.024 & -0.004 & 0.024 & 0.947 & 0.024 & -0.004 & 0.024 & 0.944 & 0.024 & -0.003 & 0.024 & 0.947\\
Expectation Consistent & \textbf{0.024} & -0.004 & 0.024 & 0.946 & 0.024 & -0.004 & 0.024 & 0.947 & \textbf{0.024} & -0.003 & 0.024 & 0.950\\
Temperature Scaling & 0.024 & -0.004 & 0.024 & 0.947 & 0.024 & -0.004 & 0.024 & 0.948 & 0.024 & -0.003 & 0.024 & 0.950\\
Isotonic Regression & 0.309 & 0.040 & 0.306 & 0.943 & 0.302 & 0.040 & 0.300 & 0.937 & 0.324 & 0.019 & 0.324 & 0.943\\
Venn Abers & 0.024 & -0.003 & 0.024 & 0.948 & 0.024 & -0.003 & 0.024 & 0.944 & 0.024 & -0.002 & 0.024 & 0.950\\
 \midrule
        &\multicolumn{12}{c}{DGP 2 (easy outcome regression / difficult propensity scores)} \\  \midrule
        DML & 0.044 & 0.035 & 0.028 & 0.749 & 0.037 & 0.025 & 0.028 & 0.844 & 0.030 & 0.013 & 0.027 & 0.919\\
Reweight DML & 0.043 & 0.033 & 0.028 & 0.770 & 0.036 & 0.023 & 0.028 & 0.861 & 0.029 & 0.011 & 0.027 & 0.931\\
Platt Scaling & 0.031 & 0.012 & 0.028 & 0.930 & 0.029 & 0.007 & 0.029 & 0.947 & \textbf{0.027} & 0.002 & 0.027 & 0.953\\
Beta Scaling & 0.030 & 0.011 & 0.029 & 0.930 & 0.029 & 0.007 & 0.029 & 0.947 & 0.027 & 0.001 & 0.027 & 0.954\\
Expectation Consistent & 0.030 & 0.010 & 0.029 & 0.940 & 0.030 & 0.008 & 0.029 & 0.944 & 0.027 & 0.001 & 0.027 & 0.953\\
Temperature Scaling & 0.033 & 0.017 & 0.028 & 0.899 & 0.030 & 0.011 & 0.028 & 0.937 & 0.027 & 0.000 & 0.027 & 0.954\\
Isotonic Regression & 0.372 & -0.158 & 0.337 & 0.923 & 0.223 & 0.003 & 0.223 & 0.945 & 0.328 & -0.069 & 0.320 & 0.937\\
Venn Abers & \textbf{0.030} & 0.010 & 0.029 & 0.937 & \textbf{0.029} & 0.003 & 0.029 & 0.953 & 0.029 & 0.001 & 0.029 & 0.951\\
 \midrule
        &\multicolumn{12}{c}{DGP 3 (difficult outcome regression / easy propensity scores)} \\  \midrule
        DML & 0.024 & -0.004 & 0.024 & 0.945 & 0.024 & -0.002 & 0.024 & 0.952 & 0.025 & -0.006 & 0.024 & 0.943\\
Reweight DML & 0.024 & -0.004 & 0.024 & 0.946 & \textbf{0.024} & -0.002 & 0.024 & 0.951 & 0.025 & -0.006 & 0.024 & 0.943\\
Platt Scaling & 0.024 & -0.003 & 0.024 & 0.952 & 0.024 & -0.003 & 0.024 & 0.949 & 0.025 & -0.003 & 0.024 & 0.947\\
Beta Scaling & 0.024 & -0.003 & 0.024 & 0.953 & 0.024 & -0.003 & 0.024 & 0.948 & 0.025 & -0.003 & 0.024 & 0.947\\
Expectation Consistent & \textbf{0.024} & -0.003 & 0.024 & 0.948 & 0.024 & -0.003 & 0.024 & 0.949 & 0.025 & -0.003 & 0.024 & 0.946\\
Temperature Scaling & 0.024 & -0.003 & 0.024 & 0.952 & 0.024 & -0.003 & 0.024 & 0.947 & \textbf{0.025} & -0.003 & 0.024 & 0.948\\
Isotonic Regression & 0.311 & 0.040 & 0.309 & 0.949 & 0.312 & 0.026 & 0.311 & 0.940 & 0.320 & -0.030 & 0.318 & 0.933\\
Venn Abers & 0.024 & -0.002 & 0.024 & 0.950 & 0.024 & -0.002 & 0.024 & 0.948 & 0.025 & -0.003 & 0.024 & 0.946\\
 \midrule
        &\multicolumn{12}{c}{DGP 4 (difficult outcome regression / difficult propensity scores)} \\  \midrule
        DML & 0.042 & 0.030 & 0.029 & 0.808 & 0.037 & 0.023 & 0.028 & 0.865 & 0.098 & 0.094 & 0.028 & 0.078\\
Reweight DML & 0.040 & 0.028 & 0.029 & 0.824 & 0.035 & 0.021 & 0.028 & 0.878 & 0.093 & 0.089 & 0.028 & 0.109\\
Platt Scaling & 0.031 & 0.010 & 0.029 & 0.936 & 0.029 & 0.005 & 0.029 & 0.951 & 0.085 & 0.081 & 0.028 & 0.186\\
Beta Scaling & 0.030 & 0.008 & 0.029 & 0.943 & \textbf{0.029} & 0.002 & 0.029 & 0.960 & 0.059 & 0.052 & 0.028 & 0.552\\
Expectation Consistent & \textbf{0.030} & 0.004 & 0.029 & 0.950 & 0.029 & 0.001 & 0.029 & 0.949 & 0.104 & 0.100 & 0.029 & 0.068\\
Temperature Scaling & 0.031 & 0.011 & 0.029 & 0.936 & 0.029 & 0.005 & 0.029 & 0.955 & 0.104 & 0.100 & 0.029 & 0.067\\
Isotonic Regression & 0.347 & -0.092 & 0.335 & 0.939 & 0.224 & -0.015 & 0.224 & 0.947 & 0.517 & 0.245 & 0.455 & 0.910\\
Venn Abers & 0.030 & 0.007 & 0.029 & 0.945 & 0.029 & 0.001 & 0.029 & 0.955 & \textbf{0.032} & 0.013 & 0.029 & 0.929\\
 \midrule
        &\multicolumn{12}{c}{DGP 5 (easy outcome regression / difficult and extreme propensity scores)} \\  \midrule
        DML & 0.084 & 0.079 & 0.027 & 0.183 & 0.082 & 0.077 & 0.028 & 0.205 & 0.061 & -0.038 & 0.048 & 0.883\\
Reweight DML & 0.079 & 0.074 & 0.028 & 0.243 & 0.078 & 0.073 & 0.028 & 0.246 & 0.051 & -0.027 & 0.044 & 0.901\\
Platt Scaling & 0.043 & 0.030 & 0.031 & 0.833 & 0.042 & 0.028 & 0.031 & 0.854 & \textbf{0.032} & 0.009 & 0.031 & 0.942\\
Beta Scaling & 0.041 & 0.027 & 0.031 & 0.863 & 0.042 & 0.028 & 0.031 & 0.860 & 0.039 & -0.004 & 0.038 & 0.948\\
Expectation Consistent & 0.045 & 0.033 & 0.031 & 0.806 & 0.047 & 0.035 & 0.031 & 0.786 & 0.060 & -0.036 & 0.048 & 0.887\\
Temperature Scaling & 0.041 & 0.026 & 0.032 & 0.863 & 0.042 & 0.029 & 0.031 & 0.854 & 0.049 & -0.025 & 0.042 & 0.898\\
Isotonic Regression & 0.412 & -0.222 & 0.347 & 0.910 & 0.162 & 0.021 & 0.161 & 0.932 & 0.421 & -0.210 & 0.365 & 0.911\\
Venn Abers & \textbf{0.038} & 0.020 & 0.033 & 0.901 & \textbf{0.037} & 0.017 & 0.033 & 0.912 & 0.034 & 0.004 & 0.034 & 0.952\\
 \midrule
        &\multicolumn{12}{c}{DGP 6 (difficult outcome regression / difficult and extreme propensity scores)} \\  \midrule
        DML & 0.071 & 0.064 & 0.031 & 0.460 & 0.056 & 0.048 & 0.030 & 0.644 & 0.113 & 0.097 & 0.056 & 0.615\\
Reweight DML & 0.065 & 0.057 & 0.031 & 0.531 & 0.053 & 0.043 & 0.030 & 0.697 & 0.100 & 0.087 & 0.051 & 0.617\\
Platt Scaling & \textbf{0.033} & 0.008 & 0.032 & 0.937 & 0.032 & -0.005 & 0.032 & 0.946 & \textbf{0.032} & 0.002 & 0.032 & 0.946\\
Beta Scaling & 0.033 & 0.007 & 0.032 & 0.943 & 0.032 & -0.006 & 0.032 & 0.947 & 0.044 & 0.017 & 0.041 & 0.937\\
Expectation Consistent & 0.034 & 0.012 & 0.032 & 0.932 & \textbf{0.031} & 0.002 & 0.031 & 0.946 & 0.112 & 0.096 & 0.057 & 0.631\\
Temperature Scaling & 0.033 & 0.004 & 0.032 & 0.946 & 0.032 & -0.005 & 0.032 & 0.950 & 0.099 & 0.086 & 0.049 & 0.605\\
Isotonic Regression & 0.388 & -0.169 & 0.349 & 0.929 & 0.164 & -0.002 & 0.164 & 0.931 & 0.497 & 0.253 & 0.427 & 0.906\\
Venn Abers & 0.034 & 0.002 & 0.034 & 0.948 & 0.036 & -0.014 & 0.033 & 0.935 & 0.037 & 0.011 & 0.036 & 0.933\\
 \midrule
    \end{tabular}
    \begin{tablenotes}[flushleft]
     \item \textit{Note:} The table shows the results for the different calibration methods for the different DGPs with $N = 8,000$. The RMSE is the root mean squared error, the bias is the average difference between the true value and the estimated value, the standard deviation is the standard deviation of the estimated values and the coverage is the coverage of the 95\% confidence interval.
    \end{tablenotes}
\end{threeparttable}
\end{adjustbox}
\end{table}
\clearpage
\subsection{Additional Empirical Results} \label{app:additional_results_empirical}

\begin{table}[h!]
    \centering
    \caption{Results for the empirical application across ten random sub-sample replications}
\label{Results_empirical_application_simulation}
    \begin{adjustbox}{max width=\textwidth}
        \begin{threeparttable}
    \begin{tabular}{l|ccc|ccc|ccc|ccc} \toprule
        Approach  &\multicolumn{3}{c}{12.5 \% of sample} &\multicolumn{3}{c}{25 \% of sample} &\multicolumn{3}{c}{50 \% of sample}&\multicolumn{3}{c}{full sample}\\  \midrule
        & Avg. ATE & Std. dev. & Min. Brier & Avg. ATE & Std. dev.& Min. Brier & Avg. ATE & Std. dev. & Min. Brier & ATE & SE & Brier score\\ \midrule
        & \multicolumn{12}{c}{\textbf{Random Forest}} \\ \midrule
        DML & -0.695 & 0.052 & 0 & -0.667 & 0.024 & 0 & -0.621 & 0.032 & 0 & -0.594 & 0.021 & 0.02871\\
Reweight DML & -0.695 & 0.052 & 0 & -0.667 & 0.024 & 0 & -0.621 & 0.031 & 0 & -0.593 & 0.023 & 0.02871\\
Platt Scaling & -0.659 & 0.052 & 90 & -0.645 & 0.022 & 80 & -0.608 & 0.031 & 50 & -0.587 & 0.030 & 0.02800\\
Beta Scaling & -0.657 & 0.052 & 0 & -0.643 & 0.023 & 0 & -0.605 & 0.030 & 10 & -0.588 & 0.031 & 0.02805\\
Expectation Consistent & -0.684 & 0.047 & 0 & -0.664 & 0.021 & 0 & -0.619 & 0.032 & 0 & -0.594 & 0.021 & 0.02871\\
Temperature Scaling & -0.684 & 0.046 & 0 & -0.664 & 0.021 & 0 & -0.619 & 0.032 & 0 & -0.594 & 0.022 & 0.02872\\
Venn Abers & -0.654 & 0.059 & 10 & -0.637 & 0.034 & 20 & -0.606 & 0.026 & 40 & -0.587 & 0.028 & 0.02797\\
 \midrule
       & \multicolumn{12}{c}{\textbf{Gradient Boosting}} \\ \midrule
        DML & -0.699 & 0.087 & 20 & -0.662 & 0.022 & 40 & -0.606 & 0.026 & 0 & -0.560 & 0.030 & 0.02749\\
Reweight DML & -0.704 & 0.079 & 20 & -0.661 & 0.022 & 40 & -0.605 & 0.026 & 0 & -0.559 & 0.032 & 0.02749\\
Platt Scaling & -0.675 & 0.102 & 0 & -0.673 & 0.016 & 0 & -0.620 & 0.034 & 0 & -0.572 & 0.023 & 0.02798\\
Beta Scaling & -0.690 & 0.057 & 60 & -0.660 & 0.018 & 40 & -0.602 & 0.028 & 80 & -0.555 & 0.034 & 0.02749\\
Expectation Consistent & -0.689 & 0.069 & 0 & -0.660 & 0.017 & 10 & -0.604 & 0.027 & 10 & -0.559 & 0.030 & 0.02749\\
Temperature Scaling & -0.691 & 0.068 & 0 & -0.660 & 0.017 & 0 & -0.603 & 0.028 & 0 & -0.558 & 0.031 & 0.02749\\
Venn Abers & -0.691 & 0.057 & 20 & -0.656 & 0.026 & 10 & -0.604 & 0.029 & 10 & -0.555 & 0.037 & 0.02754\\
 \midrule
        & \multicolumn{12}{c}{\textbf{Lasso}} \\ \midrule
        DML & -0.714 & 0.072 & 0 & -0.683 & 0.061 & 0 & -0.571 & 0.146 & 0 & -0.471 & 0.062 & 0.02777\\
Reweight DML & -0.719 & 0.074 & 0 & -0.686 & 0.059 & 0 & -0.576 & 0.100 & 0 & -0.492 & 0.050 & 0.02777\\
Platt Scaling & -0.709 & 0.068 & 0 & -0.671 & 0.039 & 50 & -0.591 & 0.029 & 0 & -0.552 & 0.024 & 0.02811\\
Beta Scaling & -0.712 & 0.063 & 40 & -0.665 & 0.046 & 40 & -0.569 & 0.040 & 70 & -0.502 & 0.044 & 0.02772\\
Expectation Consistent & -0.711 & 0.065 & 60 & -0.682 & 0.061 & 0 & -0.552 & 0.130 & 10 & -0.468 & 0.065 & 0.02777\\
Temperature Scaling & -0.712 & 0.064 & 0 & -0.682 & 0.060 & 10 & -0.558 & 0.107 & 0 & -0.477 & 0.059 & 0.02777\\
Venn Abers & -0.712 & 0.064 & 0 & -0.663 & 0.047 & 0 & -0.561 & 0.029 & 20 & -0.532 & 0.035 & 0.02780\\
 \midrule
    \end{tabular}
    \begin{tablenotes}[flushleft]
     \item \textit{Note:} The table shows the results for the empirical application for different calibration methods and different sample sizes. For each sample size, despite the full sample, the data has been drawn ten times. Avg. ATE is the average of the ATE over the different draws. Std. dev. is the standard deviation of the ATE over the different draws. Min. Brier show the percentage in which this calibrator has the lowest Brier score. The percentages do not necessarily sum to 100\% as DML and Reweight DML have the same propensity scores and thus the same Brier scores. For the full sample, the ATE is the average treatment effect, SE is the standard error of the ATE.  The Brier score is defined as $\frac{1}{N}\sum_{i = 1}^N (\hat p(X_i) - D_i)^2$.
    \end{tablenotes}
\end{threeparttable}
\end{adjustbox}
\end{table}
\end{appendices}
\end{document}